\documentclass[notitlepage, 12pt]{amsart} %
\usepackage{amsthm, amsmath, dsfont, mathrsfs, amssymb, graphicx, mathabx, amsaddr, placeins, caption, soul}
\usepackage[usenames,dvipsnames,svgnames,table]{xcolor}
\usepackage{tikz}
\usepackage{subfig} 
\usetikzlibrary{decorations.pathreplacing}
\usepackage[
    colorlinks=true,
    linkcolor=Blue,
    citecolor=Brown]{hyperref}

\newtheorem{theorem}{Theorem}
\newtheorem{lemma}{Lemma}
\newtheorem{proposition}{Proposition}
\newtheorem{corollary}{Corollary}
\newtheorem{definition}{Definition}

\newtheorem{remark}{Remark}
\newtheorem{claim}{{\sc Claim}}[section]

\newcolumntype{C}[1]{>{\centering\arraybackslash}p{#1}} 

\usepackage{geometry}
\title[Strategically Simple Mechanisms]{Strategically Simple Mechanisms}\thanks{We are indebted to the editor and four anonymous referees for their comments and suggestions that substantially improved the paper. We are grateful to Gabriel Carroll, Yi-Chun Chen, Johannes H\"orner, Heng Liu, Alessandro Pavan, Xianwen Shi, Satoru Takahashi, and participants at various seminars and conferences for helpful discussions.}

\author[Tilman B\"orgers]{Tilman B\"orgers}\address{Department of Economics, University of Michigan}

\author[Jiangtao Li]{Jiangtao Li}\address{School of Economics, Singapore Management University}

\date{December 1, 2018.}

\begin{document}

\begin{abstract}
	
We define and investigate a property of mechanisms that we call ``strategic simplicity,'' and that is meant to capture the idea that, in strategically simple mechanisms, strategic choices require limited strategic sophistication. We define a mechanism to be strategically simple if choices can be based on first-order beliefs about the other agents' preferences and first-order certainty about the other agents' rationality alone, and there is no need for agents to form higher-order beliefs, because such beliefs are irrelevant to the optimal strategies. All dominant strategy mechanisms are strategically simple. But many more mechanisms are strategically simple. In particular, strategically simple mechanisms may be more flexible than dominant strategy mechanisms in the bilateral trade problem and the voting problem.

\end{abstract}

\maketitle

\thispagestyle{empty}

\clearpage

\thispagestyle{empty}

\null\null\null\null\null\null\null\null

\tableofcontents

\clearpage

\setcounter{page}{1}

\section{Introduction}\label{sec:intro}\smallskip

In mechanism design it seems useful to distinguish mechanisms in which agents face a straightforward choice problem from mechanisms that require agents to engage in complex thinking if they want to determine their optimal choices. In some cases, the mechanism designer might prefer a mechanism of the former type. There are several conceivable reasons for such a preference. For example, the mechanism designer might not be able to predict the behavior of agents who are not capable of complex thinking. Or the mechanism designer might find it desirable that the outcomes of a mechanism don't depend much on the cognitive abilities of the agents. Of course, one can also imagine settings in which the mechanism designer prefers mechanisms that make it hard for agents to find an optimal strategy.\medskip

In this paper, we introduce a property of mechanisms that is intended to capture the idea that strategic choices are simple in a particular way: agents don't need to have much strategic sophistication to determine their optimal strategies. Here we mean by ``strategic sophistication'' the ability to reason about the other agents' preferences, the other agents'  beliefs, their beliefs about beliefs, etc. That forming higher-order beliefs is difficult seems plausible from everyday experience, and there is also some recent experimental evidence that points into this direction; see, for example, Alaoui and Penta \cite{AlaouiPenta2017}.\medskip

There are, of course, many other dimensions to simplicity. For example, the sheer size of the strategy space and the complexity of the mapping that assigns outcomes to strategy combinations may make a mechanism difficult to understand.  One may also consider the computational complexity of finding optimal strategies. In settings in which the agents interact repeatedly, one may consider how easy it is for the agents to learn through repeated play.\footnote{Mathevet \cite{Mathevet2010} introduces an approach to construct supermodular mechanisms. Supermodular mechanisms are desirable in settings in which the agents interact repeatedly and in settings in which learning and adjusting are important.} The purpose of this paper is to isolate just one dimension of simplicity, and to consider one possible formalization of this dimension.\medskip

One class of mechanisms which require little strategic sophistication is the class of dominant strategy mechanisms.\footnote{A dominant strategy mechanism is a mechanism in which each agent has a dominant strategy regardless of her preference. We use the phrase ``dominant strategy'' in the sense in which it is used in mechanism design theory, that is, a strategy that is optimal regardless of what the other agents do. This is slightly different from ``weakly dominant'' or ``strictly dominant'' strategies as these terms are defined in game theory. Dominant strategy mechanisms are also called ``strategy-proof'' mechanisms in the literature.} In dominant strategy mechanisms, agents need not think at all about the motives of the other agents, or the other agents' rationality. This is because each agent has at least one strategy that is optimal regardless of what the other agents do, and she can just choose such a strategy.\medskip

For many mechanism design problems, the class of dominant strategy mechanisms is quite small, and only includes mechanisms that are rather unattractive for a mechanism designer who wants to maximize, say,  revenue, or welfare.\footnote{See the examples in Chapter 4 of B\"orgers \cite{Borgers2015}.} We therefore introduce in this paper a new class of mechanisms that includes, but is strictly larger than, the set of dominant strategy mechanisms. We call the mechanisms in this class ``strategically simple.'' We argue that for mechanisms in this new class the strategic sophistication needed to find optimal strategies is quite limited. Our results show that in applications the set of strategically simple mechanisms includes mechanisms that are more attractive to a mechanism designer concerned with fairness, efficiency, or revenue, than dominant strategy mechanisms.\medskip

The following example illustrates our idea. Consider the well-known problem of designing a mechanism that allows a seller of an indivisible object to trade with one potential buyer. Both agents have quasi-linear preferences. It is known from Hagerty and Rogerson \cite{HagertyRogerson1987} that the only dominant strategy mechanisms that satisfy ex post budget balance and individual rationality are posted price mechanisms. In a posted price mechanism, the designer chooses a (possibly random) price, without taking into account any of the agents' private information. The outcome depends on agents' private information only through their decision to trade, or not to trade, at the price proposed by the mechanism designer. Trade comes about only when both agents agree. Obviously, this is a rather unappealing mechanism for a welfare maximizing mechanism designer.\medskip

Now consider an alternative mechanism that we call ``price cap mechanism.'' The mechanism designer sets a price cap. The seller can refuse to trade, or choose a price less than or equal to the mechanism designer's price cap, and indicate that he is willing to trade at this price, or at a lower price. If the seller is willing to trade, then the buyer can decides whether or not to trade at the price chosen by the seller. Trade takes only place if both agents agree to trade.\medskip

 Whether or not to reduce the price, and how far to reduce the price, depends on the seller's belief about the buyer's willingness to pay. But, regardless of her belief, the seller will never reduce the price below her reservation value, and the buyer will never agree to trade if the potentially reduced price is above his willingness to pay. In comparison to the posted price mechanism, this mechanism facilitates more efficient trade.\footnote{This mechanism was discussed in B\"orgers and Smith \cite{BorgersSmith2012}.}\medskip

In the price cap mechanism, the buyer faces a straightforward choice problem. The buyer agrees to trade if and only if his willingness to pay is weakly higher than the price offered. The seller's problem is arguably not too complicated either. If she believes that the buyer accepts the trade if and only if his willingness to pay is weakly higher than the price offered, then all that she needs to do is to consider her belief about the buyer's willingness to pay. This problem is equivalent to the standard monopoly problem with a price ceiling, as taught in undergraduate microeconomics. For any belief that the seller might have, it is a straightforward optimization problem. Our formal definition of strategic simplicity will imply that the price cap mechanism is strategically simple.\medskip

On the other hand, the double auction described in Chatterjee and Samuelson \cite{ChatterjeeSamuelson1983} is, in our terminology, not strategically simple. To see why, note that in the double auction, the seller has to form her belief about the price that the buyer offers. Ideally, she would like to ask for a price that is as close as possible to, but not above the price that is offered by the buyer, provided that this price is above her reservation value. But to form her belief about the price that the buyer offers, presumably the seller first has to form her belief about the buyer's belief about the seller's reservation value. Similarly, the buyer has to form his belief about the seller's belief about the buyer's willingness to pay. Potentially, infinitely many layers of such beliefs matter. Mechanisms that require  of agents this level of depth of thinking will, in our terminology, not be strategically simple.\medskip

Motivated by this example, we define in this paper a mechanism to be strategically simple if optimal choices can be determined using first-order beliefs alone, and there is no need for agents to form higher-order beliefs because such beliefs are irrelevant to the optimal strategies. Here, we are referring to beliefs about the other agents' utility functions and rationality. Thus, a ``first-order belief''  of agent $i$ is agent $i$'s belief about the other agents' ($j\neq i$) utility functions, and about the other agents' rationality. ``Higher-order beliefs'' are, for example, agent $i$'s belief about agent $j$'s belief about agent $i$'s utility function, and about agent $j$'s belief about agent $i$'s rationality. We shall call a mechanism strategically simple if for each agent $i$, her belief about the other agents' ($j\neq i$) utility functions, combined with certainty that the other agents are rational, imply which choices are optimal for agent $i$.\footnote{It is, of course, somewhat arbitrary to restrict attention to mechanisms for which only first-order beliefs matter. We might instead allow first and second-order beliefs to matter, for example. We discuss such variations  in Section \ref{sec:discussion}.}\medskip

Our definition of strategic simplicity allows for the possibility that only some subset of all utility functions and some subset of all first-order beliefs are considered relevant for the determination whether or not a mechanism is strategically simple. Our main result shows that, under a ``richness" condition on the domain of relevant utility functions and first-order beliefs, strategic simplicity is equivalent to a ``local dictatorship'' property. The richness condition will be formally defined later, but we emphasize that it is much weaker than the requirement that all possible utility functions and all possible first-order beliefs are in the domain. In contrast with (classical)  dictatorship, local dictatorship means, roughly speaking, that there is some agent who dictates the outcome if we restrict attention for every agent to certain subset of her strategy set. The identity of the dictator may depend on the subsets that we consider. Every dictatorship mechanism is a local dictatorship mechanism, but there are many more local dictatorship mechanisms than dictatorship mechanisms, including mechanisms that are far from what in everyday language is called a ``dictatorship.'' For example, a voting mechanism in which one agent selects two alternatives from a larger set of several alternatives, and the other agents then vote over those two alternatives using majority voting, is in our language a local dictatorship.\medskip

Our characterization result suggests a natural division of strategically simple mechanisms into two categories: mechanisms in which there is some agent who is a local dictator at all restrictions that we consider, and mechanisms in which this is not the case. We shall call the former ``type 1 strategically simple mechanisms," and the latter ``type 2 strategically simple mechanisms.'' Type 1 strategically simple mechanisms are easy to characterize. One can think of type 1 strategically simple mechanisms as ``delegation mechanisms:'' the mechanism designer delegates the choice of the mechanism to a ``delegate,'' who chooses a mechanism from a given set of dominant strategy mechanisms that the designer has specified. Then the other agents play this dominant strategy mechanism. The delegate's choice will depend on her first-order belief, while the other agents' choices don't require any belief formation. Type 1 strategically simple mechanisms include, as a subclass, all dominant strategy mechanisms. The price cap mechanism presented above is also a type 1 strategically simple mechanism.\medskip

Type 2 strategically simple mechanisms are harder to characterize in general. Instead, after presenting general results, we shall analyze in this paper two applications. In these applications we will be able to characterize not only type 1 but also, under some assumptions, type 2 strategically simple mechanisms.\medskip

The first application that we study is the bilateral trade problem. We show that the type 1 strategically simple mechanisms are those in which one agent proposes terms of trade, and the other agent accepts or rejects. We then show that there are no type 2 strategically simple mechanisms in the bilateral trade environment. Thus, we fully characterize strategically simple mechanisms in the bilateral trade environment.\medskip

 The second application that we study is the voting problem. In the voting environment, the class of strategically simple mechanisms is much larger than the class of dominant strategy mechanisms. By the celebrated Gibbard-Satterthwaite \cite{Gibbard1973, Satterthwaite1975} Theorem, in the voting environment as we define it here, a mechanism has dominant strategies if and only if it is dictatorial. There are many more strategically simple voting mechanisms. In this paper, we characterize type 1 strategically simple voting mechanisms and we characterize type 2 strategically simple voting mechanisms when there are two agents and three alternatives. For this special case we show that there is a type 2 strategically simple mechanism that is anonymous,  that is, that treats both voters equally. By contrast, no type 1 strategically simple mechanism, and, in particular, no dominant strategy mechanism, is anonymous in this setting.\medskip

The paper is organized as follows. Section \ref{sec:definitions} contains the formal definition of strategic simplicity, and Section \ref{sec:examples} provides an example to illustrate this notion. Section \ref{sec:characterization} contains our characterization result of strategically simple mechanisms under a richness condition on the domain of utility functions and beliefs. Sections \ref{sec:bilateral-trade} and \ref{sec:voting} consider the applications of our main results in the bilateral trade problem and the voting problem. Section \ref{sec:literature} reviews the related literature. Section \ref{sec:discussion} is a discussion of open questions.

\section{Definitions} \label{sec:definitions}

\smallskip

There are $n$ agents: $i \in I = \{1,2,\ldots,n\}$ and a finite set $A$ of outcomes. A mechanism consists of a finite strategy set $S_i$ for each agent $i$, and an outcome function $g: \bigtimes_{i\in I} S_i \rightarrow A$ that describes for each choice of strategies which outcome will result. We define $S\equiv \bigtimes_{i\in I} S_i$ with generic element $s$, and, for every $i\in I$, we define $S_{-i}\equiv\bigtimes_{j\neq i}S_j$ with generic element $s_{-i}$. We assume that there are no duplicate strategies: for every $i\in I$, for all $s_i,s_i' \in S_i$ with $s_i \neq s_i'$, there is some $s_{-i}\in S_{-i}$ such that $g(s_i,s_{-i})\neq g(s_i',s_{-i})$.\medskip

A von Neumann-Morgenstern (vNM) utility function of agent $i$ is a function $u_i: A \rightarrow \mathds{R}$.  We define $\mathcal{U}$ to be the set of all utility functions such that: $u_i(a) \neq u_i(a')$ whenever $a\neq a'$, $\min_{a \in A} u_i(a) = 0$, and $\max_{a \in A} u_i(a) = 1$. Thus we rule out indifferences and normalize utility. This simplifies arguments below. We write $u\equiv (u_1,u_2,\ldots,u_n)$ and $u_{-i}\equiv (u_j)_{j\neq i}$.\medskip

For every agent $i$, there is a non-empty and Borel-measurable set $\bf{U}_i\subseteq \mathcal{U}$ of utility functions that are possible utility functions of agent $i$. We allow for the possibility that $\mathbf{U}_i\neq \mathcal{U}$ to be able to capture assumptions such as the assumption that agents' utility functions are quasi-linear. We define $\mathbf{U}\equiv \bigtimes_{i\in I} \mathbf{U}_i$, and, for every $i\in I$, we define $\mathbf{U}_{-i} \equiv \bigtimes_{j\neq i} \mathbf{U}_j$.\medskip

For a given mechanism, for every $i$ and every $u_i\in \mathbf{U}_i$, we denote by $U\!D_i(u_i)$ the set of all strategies that are not weakly dominated for agent $i$ with utility function $u_i$, where weak dominance may be by a pure or by a mixed strategy. If $u\in \mathbf{U}$, we define $U\!D(u)\equiv \bigtimes_{i\in I}U\!D_i(u_i)$. For every $i \in I$ and every $u_{-i} \in \mathbf{U}_{-i}$, we define $U\!D_{-i}(u_{-i})\equiv \bigtimes_{j\neq i} U\!D_j(u_j)$. To avoid tedious detail, we assume that for every agent $i\in I$ and every strategy $s_i\in S_i$, there is at least some $u_i\in \mathbf{U}_i$ such that $s_i\in U\!D_i(u_i)$.\medskip

A ``utility belief'' $\mu_i$ of agent $i$ is a Borel probability measure on $\mathbf{U}_{-i}$. We interpret $\mu_i$ as agent $i$'s ``first-order'' belief. Higher-order beliefs would be beliefs about other agents' beliefs about utility functions, etc. As indicated in the introduction, we want to focus on mechanisms in which higher-order beliefs play no role. Therefore, we don't formally define such beliefs here.\medskip

For any finite set (or Borel subset of a finite dimensional Euclidean space) $X$, we shall denote by $\Delta(X)$ the set of all (Borel) probability measures on $X$. The set of all possible utility beliefs of agent $i$ is some non-empty subset $\mathbf{M}_i$ of $\Delta(\mathbf{U}_{-i})$. We allow for the possibility that $\mathbf{M}_i\neq \Delta(\mathbf{U}_{-i})$ to be able to capture assumptions such as the assumption that every agent believes that the other agents' utility functions are stochastically independent.  We define $\mathbf{M}\equiv \bigtimes_{i\in I} \mathbf{M}_i$, and, for every $i\in I$, we define $\mathbf{M}_{-i} \equiv \bigtimes_{j\neq i} \mathbf{M}_j$, and we denote typical elements of these sets by $\mu$ and $\mu_{-i}$ respectively.\medskip

 A ``strategic belief'' $\hat \mu_i$ of agent $i$ is a probability measure on $S_{-i}$: $\hat \mu_i\in \Delta(S_{-i})$. Strategic beliefs are needed for agents to determine expected utility maximizing strategies. The next definition will describe how agents may derive a strategic belief from a utility belief. We assume that agents are certain that the other agents do not play weakly dominated strategies. Then, loosely speaking, a strategic belief can be obtained from a given utility belief by dividing the probability assigned to any utility function $u_j$ $(j\neq i)$ in some arbitrary way among the not weakly dominated strategies of agent $j$ with utility function $u_j$. We call a strategic belief that can be derived in this way from a utility belief ``compatible with the utility belief." Obviously, for a given utility belief, there may be many compatible strategic beliefs.  We formally define the compatibility of strategic beliefs with utility beliefs as follows:

\begin{definition} \label{compatible}
A strategic belief $\hat \mu_i$ is ``compatible with a utility belief $\mu_i$'' if there is a probability measure $\nu_i$ on $S_{-i}\times \mathbf{U}_{-i}$ that has support in $$\bigtimes\limits_{j\neq i}\left\{(s_j,u_j)\in S_j \times \mathbf{U}_j | s_j\in U\!D_j(u_j)\right\}$$ and that has marginals $\hat \mu_i$ on $S_{-i}$ and $\mu_i$ on $\mathbf{U}_{-i}$.
\end{definition}

In this definition, $\nu_i$ is agent $i$'s joint belief about strategies and utility functions of the other agents. Agent $i$'s certainty that the other agents don't play weakly dominated strategies is captured by the support restriction in Definition \ref{compatible}. The belief $\nu_i$ must also reflect the given utility belief $\mu_i$ of agent $i$, that is, $\nu_i$'s marginal on $\mathbf{U}_{-i}$ must be $\mu_i$. The marginal on $S_{-i}$ is then a compatible strategic belief.  We denote the set of all strategic beliefs that are compatible with a given utility belief $\mu_i$ by $\mathcal{M}_i(\mu_i)$.\medskip

Given a utility function $u_i\in \mathbf{U}_i$ and a strategic belief $\hat \mu_i\in \Delta(S_{-i})$ of agent $i$, we denote by $B\!R_i(u_i,\hat \mu_i)$ the set of all strategies in $U\!D_i(u_i)$ that maximize expected utility in $S_i$.\medskip

We are now ready to provide the key definition of this paper.

\begin{definition} \label{simple}
	
A mechanism is ``strategically simple'' with respect to $\mathbf{U}$ and $\mathbf{M}$\footnote{Whether a given mechanism is strategically simple or not depends on the domain of utility function $\mathbf{U}$ and the domain of first-order beliefs $\mathbf{M}$ that we study. For simplicity, we sometimes drop the quantifier ``with respect to $\textbf{U}$ and $\textbf{M}$'' when there is no confusion.} if for every agent $i$, every  utility function $u_i\in \mathbf{U}_i$, and every utility belief $\mu_i\in \mathbf{M}_i$,
\[
\bigcap\limits_{\hat \mu_i \in \mathcal{M}_i(\mu_i)} B\!R_i(u_i,\hat \mu_i)\neq \emptyset.
\]

\end{definition}

What we require here for every agent $i$, every utility function $u_i$ of agent $i$, and every utility belief $\mu_i$ of agent $i$, is that agent $i$ has at least one strategy that maximizes expected utility regardless of which compatible strategic belief $\hat \mu_i$ agent $i$ picks. Thus, there is no need for agent $i$ to try to distinguish more plausible from less plausible compatible strategic beliefs. If that was necessary, it may be helpful for agent $i$  to form higher-order beliefs. But if a mechanism is strategically simple, there is no benefit to agent $i$ from forming higher-order beliefs.

\begin{remark}
{\emph If agent $i$ with utility function $u_i$ has a weakly dominant strategy, then this strategy is included in $B\!R_i(u_i,\hat \mu_i)$ regardless of what $\hat \mu_i$ is, and the intersection referred to in Definition \ref{simple} is non-empty because it includes the dominant strategy. Dominant strategy mechanisms, in which all agents for all utility functions have dominant strategies, are therefore trivially ``strategically simple.''}
\end{remark}

\begin{remark}
{\emph In the definition of compatible strategic beliefs (Definition \ref{compatible}), we allow agents to incorporate correlations into their beliefs about the other agents' strategies that go beyond the correlations implied by correlations in utility beliefs and the requirement that not weakly dominated strategies are played. As is well known, without allowing for arbitrary correlations in strategic beliefs, the equivalence between not dominated strategies and expected utility maximizing strategies need not hold. This equivalence is invoked in our proofs. We have not pursued how our analysis would change if we did not allow such correlations.}
\end{remark}

\begin{remark}
{\emph One might conjecture that our definition of strategic simplicity of a mechanism is equivalent to the requirement that the mechanism, appropriately transformed into an incomplete information game, is dominance solvable in two steps, where the first step eliminates weakly dominated strategies and the second step eliminates strictly dominated strategies. A statement very similar to this is indeed true, except that we do not require that after two steps of the elimination procedure for every type\footnote{Here, a ``type'' must be interpreted as a pair consisting of a utility function and a utility belief.} a single strategy is left over, but rather, that every type has at least one strategy that is optimal for every strategic belief about the other types' remaining strategies. This perspective on our definition facilitates comparison with the requirement that every type have a dominant strategy, because that just means that the elimination of dominated strategies stops after one step. Nonetheless, we have found the definition used in this paper in terms of first-order beliefs expresses more directly the intuitive idea on which this paper is based.}
\end{remark}

Often a mechanism designer's interest is not in the mechanism itself, but in the outcomes that result when agents pick their strategies rationally. For strategically simple mechanisms, which strategy maximizes expected utility will depend not only on an agent's utility function, but also on this agent's utility belief. It is therefore natural to focus on correspondences that map utility functions and utility beliefs into sets of outcomes. We call such correspondences ``outcome correspondences.''

\begin{definition}
The ``outcome correspondence'' associated with a mechanism that is strategically simple with respect to $\textbf{U}$ and $\textbf{M}$ is the correspondence:
\[F: \mathbf{U}\times \mathbf{M} \twoheadrightarrow A\]
defined by:
\[F(u,\mu)\equiv g\left(\bigtimes_{i\in I}\left(\bigcap\limits_{\hat \mu_i \in \mathcal{M}_i(\mu_i)} B\!R_i(u_i,\hat \mu_i)\right)\right) \mbox{ for all } (u,\mu)\in \mathbf{U}\times \mathbf{M}.\]
\end{definition}\medskip

The following definition will be useful:

\begin{definition}
Two strategically simple mechanisms are ``equivalent'' if the outcome correspondences associated with these two mechanisms are the same.
\end{definition}

The literature already contains the concept of a ``social choice correspondence.'' Social choice correspondences are similar to ``outcome correspondences,'' except that their domain consists of profiles of utility functions (or preferences) only, and does not include profiles of first-order beliefs. Focusing on utility functions in the domain seems natural if one gives the correspondence a normative interpretation, as a reflection of the outcomes that the mechanism designer regards desirable. Here, however, we give our correspondence a positive interpretation: it is a description of the end result of a given mechanism. By including the first-order beliefs in the domain, we give a more detailed description of the consequences resulting from rational choice in a given mechanism than we would obtain if only preference profiles were in the domain.\medskip

Our definition of outcome correspondences assumes that for any given utility function $u_i$ and utility belief $\mu_i$ agent $i$ will only choose strategies from the set  $$\bigcap_{\hat \mu_i \in \mathcal{M}_i (\mu_i)} B\!R_i(u_i,\hat \mu_i).$$ This implies that an agent $i$ will not choose a strategy if it is a best response to only some strategic beliefs compatible with the agent's given utility belief, but not to all compatible strategic beliefs. This assumption is in the spirit of our basic hypothesis that agents find it costly to refine their strategic beliefs, beyond making it compatible with their utility belief, and will avoid doing so if they can.\medskip

One can interpret singleton-valued outcome correspondences as direct mechanisms in which agents report their utility functions and their utility beliefs.  Using this interpretation, one can then ask whether a revelation principle holds, i.e.: If a singleton-valued outcome correspondence is implemented by a strategically simple mechanism, is then the direct mechanism defined by the outcome correspondence itself a strategically simple mechanism, and is truth telling an optimal strategy for all utility functions and first-order beliefs, regardless of higher-order beliefs, in this mechanism? Unfortunately, a technical problem that we encounter when asking this question is that we have defined strategically simple mechanisms only for the case that a mechanism has a finite strategy set for each agent, whereas we have allowed the sets of pairs of utility functions and beliefs to be infinite, and thus the direct mechanism may have infinite strategy sets. This problem is bypassed if attention is restricted to the case of finite $\mathbf{U} \times \mathbf{M}$. In this case, one can then verify that the revelation principle as described above holds. Some of our analysis below is specifically about the case of infinite $\mathbf{U} \times \mathbf{M}$ and finite mechanisms, and therefore the revelation principle will not play an important role in our analysis, in contrast to the conventional theory of mechanism design.\medskip

The formal framework developed in this section suggests two possible focuses for our analysis: the characterization of strategically simple mechanisms and the characterization of outcome correspondences that are associated with strategically simple mechanisms. We find it convenient to focus on mechanisms themselves. But we shall explain some of the implications of our results for outcome correspondences.

\section{An Example} \label{sec:examples}

\smallskip

In this section, we provide an example to illustrate the mechanics of our definition of strategically simple mechanisms. Consider the mechanism in Figure \ref{fig:example}.\footnote{This mechanism is of special importance in Section \ref{sec:voting}.} In this mechanism, agent 1 (he, the row player) and agent 2 (she, the column player) collectively choose an outcome from three alternatives $\{a,b,c\}$. In what follows, we shall apply Definition \ref{simple} and show that this mechanism is strategically simple with respect to the full domain $\mathcal{U}^2 \times (\Delta(\mathcal{U}))^2.$

\begin{figure}[h]
\begin{center}
\begin{tabular}{|c|c|c|c|c|}\hline
& $L$ & $C1$ & $C2$ & $R$ \\ \hline
$T$ & $a$ & $a$&$a$&$a$ \\ \hline
$M1$ & $a$ & $b$&$a$&$b$ \\ \hline
$M2$ & $a$ & $b$&$c$&$b$ \\ \hline
$B$ & $a$ & $b$&$c$&$c$ \\ \hline
\end{tabular}
\end{center}
\caption{A strategically simple mechanism} \label{fig:example}
\end{figure}

We have to show that each agent can find an expected utility maximizing strategy on the basis of first-order belief alone. This is obvious for both agents if they rank $a$ or $b$ highest because then they have weakly dominant strategies. Also if $u_1(c)>u_1(b)>u_1(a)$ agent 1 has a weakly dominant strategy, as has agent 2 if $u_2(c)>u_2(a)>u_2(b)$.\medskip

Thus, we have only two cases with multiple not weakly dominated strategies: $T$ and $B$ are not weakly dominated for agent 1 if $u_1(c)>u_1(a)>u_1(b)$, and $C2$ and $R$ are not weakly dominated for agent 2 if $u_2(c)>u_2(b)>u_2(a)$. For any first-order belief $\mu_1$ of agent 1, the set of compatible strategic beliefs $\mathcal{M}_1 (\mu_1)$ is:\footnote{In what follows we write $\mu_1 (abc)$ to denote the probability that agent 1 attaches to agent 2's utility satisfying $u_2(a)>u_2(b)>u_2(c)$, and we use analogous notation for the probabilities of other orderings of $a$, $b$, and $c$.}

\begin{align*}
\mathcal{M}_1 (\mu_1) = \Big\{ \hat{\mu}_1: \enspace & \hat{\mu}_1 (L) = \mu_1 (abc) + \mu_1 (acb), \\
& \hat{\mu}_1 (C1) = \mu_1 (bac) + \mu_1 (bca), \\
& \hat{\mu}_1 (C2) = \mu_1 (cab) + x, \\
& \hat{\mu}_1 (R) = \mu_1 (cba) - x, \\
& \mbox{ where } 0 \leq x \leq \mu_1 (cba) \Big\}.
\end{align*}

The multiplicity of compatible strategic beliefs is because for agent 2 with utility $u_2(c)>u_2(b)>u_2(a)$, both $C2$ and $R$ are not weakly dominated. Therefore, agent $1$ cannot pin down the strategic belief on the basis of his first-order belief alone. Nevertheless, in this mechanism, agent 1 with utility $u_1(c)>u_1(a)>u_1(b)$ only cares about the probability of agent 2 playing $C1$ and the total probability of agent 2 playing $C2$ and $R$. Hence, despite the multiplicity of compatible strategic beliefs, agent 1 with utility $u_1(c)>u_1(a)>u_1(b)$ has a strategy that is a best response to any of the compatible strategic beliefs. In other words, agent 1 can determine his expected utility maximizing strategy on the basis of his first-order belief alone. Using similar arguments, one can check that agent 2 with preference $cba$ can determine her expected utility maximizing strategy on the basis of her first-order belief alone.

\section{Characterization} \label{sec:characterization}

\smallskip

We now provide a characterization result for strategically simple mechanisms under a richness assumption regarding the sets of relevant utility functions and first-order beliefs. We denote by $\mathcal{R}$ the set of all reflexive, complete, transitive, and anti-symmetric preference relations over the set of alternatives $A$. A generic element of $\mathcal{R}$ will be denoted by $R_i$, where the index refers to agent $i$, and we denote by $P_i$ the asymmetric part of $R_i$. Every utility function $u_i \in \mathcal{U}$ induces a preference relation in the following way: $a R_i b \Leftrightarrow u_i(a) \geq u_i(b)$ and $a P_i b \Leftrightarrow u_i(a) > u_i(b)$. We denote by $\mathcal{U}(R_i)$ the set of all utility functions in $\mathcal{U}$ that induce $R_i$.\medskip

Next, we extend the notion of weak dominance to the case that only pure strategy dominance is considered. In this case, only the  preference $R_i$ induced by agent $i$'s utility function $u_i$ matters.

\begin{definition}
Let $R_i \in \mathcal{R}$. A strategy $s_i \in S_i$ is ``weakly dominated given $R_i$'' if there is another strategy $\hat s_i\in S_i$ such that
\[ g(\hat s_i,s_{-i}) \, R_i  \, g(s_i,s_{-i})\]
for all $s_{-i}\in S_{-i}$ and
\[g(\hat s_i,s_{-i}) \, P_i \, g(s_i,s_{-i})\]
for some $s_{-i}\in S_{-i}$.
\end{definition}

For any $R_i \in \mathcal{R}$, we denote by $U\!D_i(R_i)\subseteq S_i$ the set of strategies of agent $i$ that are not weakly dominated given $R_i$. For any $R = (R_1, R_2, \ldots, R_n) \in \mathcal{R}^n$, we define $U\!D_{-i}(R_{-i}) \equiv \bigtimes_{j\neq i} U\!D_j(R_j)$ for every $i \in I$.\medskip

\begin{theorem} \label{thm:MainR}
Suppose for every agent $i$ there is a non-empty set $\mathcal{R}_i\subseteq \mathcal{R}$ such that $\mathbf{U}_i=\bigcup_{R_i \in \mathcal{R}_i}\mathcal{U}(R_i)$, and suppose $\mathbf{M}_i = \Delta(\mathbf{U}_{-i})$ for all $i\in I$. Then a mechanism is strategically simple with respect to $\textbf{U}$ and $\textbf{M}$ if and only if for every $R \in \bigtimes_{i\in I} \mathcal{R}_i$  there is an agent $i^*\in I$ such that for every strategy $s_{i^*} \in U\!D_{i^*}(R_{i^*} )$ there is an alternative $a\in A$ such that:
\[ g(s_{i^*},s_{-i^*})=a \mbox{ for all } s_{-i^*} \in U\!D_{-i^*}(R_{-i^*}).\]
\end{theorem}

In words, the condition that is necessary and sufficient for strategically simple mechanisms says the following. Whenever we fix a vector of preferences $(R_1, R_2, \ldots, R_n) \in \bigtimes_{i \in I} \mathcal{R}_i$ and consider the mechanism restricted to the strategy sets $U\!D_i(R_i)$ for all $i\in I$, then, in the restricted mechanism, some agent $i^*$ is a dictator. That is, for each of the alternatives that are possible when agents choose their strategies from $U\!D_i(R_i)$, agent $i^*$ has a strategy that enforces that alternative if all other agents choose from $U\!D_i(R_i)$, and each of agent $i^*$'s strategies enforces some alternative. We call agent $i^*$ a ``local dictator,'' because in the restricted game agent $i^*$ dictates which alternative is chosen.\medskip

The theorem applies only to certain domains of utility functions and beliefs. Specifically, the theorem assumes that for each agent the set of relevant utility functions is the set of all utility functions that induce some preference from a given set of preferences, and that for each agent the relevant beliefs are all beliefs that have support in the set of considered utility functions. We thus allow restricted domains of strategic simplicity, but domains that still satisfy strong ``richness'' conditions.  In some settings, such as voting settings, these assumptions may be plausible, whereas in other settings, they may be less desirable. For example, when the allocation of money is part of the specification of alternatives, our assumption on the set of utility functions rules out that only risk neutral agents are considered, even though that is a popular case in the mechanism design literature. The assumption on the set of relevant beliefs rules out that each agent regards the other agents' preferences as stochastically independent. Our proof of Theorem \ref{thm:MainR} makes strong use of these assumptions, and we have not yet found useful results for smaller domains.\medskip

Theorem \ref{thm:MainR} characterizes strategically simple mechanisms in terms of the local dictatorship property. The local dictatorship property is useful in several aspects. First, it provides a powerful tool to check whether a given mechanism is strategically simple. Second, the local dictatorship property can be used to establish several further properties of strategically simple mechanisms. These properties are contained in Appendix \ref{app:structure}. Third, we use the local dictatorship property to further study strategically simple mechanisms in two applications. In the bilateral trade environment that we study in Section \ref{sec:bilateral-trade}, we fully characterize the class of all strategically simple mechanisms. In the voting environment that we study in Section \ref{sec:voting}, we fully characterize the class of all strategically simple mechanisms when there are two agents and three alternatives. All these results build on the local dictatorship property that we established in Theorem \ref{thm:MainR}.\medskip


Furthermore, the local dictatorship property implies that certain outcome correspondences cannot be associated with any strategically simple mechanism. Loosely speaking, the set of alternatives can depend on at most one agent's vNM utilities and utility beliefs when we hold a preference profile $R$ fixed. We formalize this property in the following definition. We say that a profile of utility functions $u_{-i}$ induces a profile of preference relations $R_{-i}$ if for every $j \in I \setminus \{i\}$, $u_j$ induces $R_j$, and  that a profile of utility functions $u$ induces a profile of preference relations $R$ if for every $i \in I$, $u_i$ induces $R_i$.

\begin{definition}
Let $i\in I$ and $R\in \mathcal{R}^n$. An outcome correspondence $F: \mathbf{U} \times \mathbf{M} \twoheadrightarrow A$ is ``non-responsive to the vNM utilities and utility beliefs of agents $j\neq i$ at $R$'' if, whenever $u_i\in \mathbf{U}_i$ induces $R_i$,  $u_{-i},\hat u_{-i}\in \mathbf{U}_{-i}$ both induce $R_{-i}$, $\mu_i\in \mathbf{M}_i$, and $\mu_{-i}, \hat \mu_{-i}\in \mathbf{M}_{-i}$, then:
\[F((u_i,u_{-i}),(\mu_i, \mu_{-i}))=F((u_i,\hat u_{-i}), (\mu_i, \hat \mu_{-i})).\]
\end{definition}

In words, the outcome correspondence is non-responsive to agents $j\neq i$ at $R$ if, as long as agents' utility functions represent the preferences in $R$, then the von Neumann Morgenstern utility functions and beliefs of agents $j\neq i$ have no impact on the set of outcomes. The following result follows directly from Theorem 1. We don't give a formal proof.

\begin{corollary}
Suppose for every agent $i$ there is a non-empty set $\mathcal{R}_i\subseteq \mathcal{R}$ such that $\mathbf{U}_i=\bigcup_{R_i \in \mathcal{R}_i}\mathcal{U}(R_i)$, and suppose $\mathbf{M}_i = \Delta(\mathbf{U}_{-i})$ for all $i\in I$. If an outcome correspondence $F: \mathbf{U} \times \mathbf{M} \twoheadrightarrow A$ can be associated with a mechanism that is strategically simple, then for every preference profile $R\in \mathcal{R}^n$ there is some agent $i^*$ such that the correspondence $F$ is non-responsive to  the vNM utilities and utility beliefs of agents $j\neq i^*$ at $R$.

\end{corollary}

Agent $i^*$ in this corollary is obviously the local dictator at $R$. This corollary implies, for example, that it is impossible to find a strategically simple mechanism that on its whole domain implements alternatives that maximize ex post utilitarian welfare, that is, the sum of agents' utilities.\medskip

Before we move on to the applications, we now partition the set of all strategically simple mechanisms on domains that satisfy the assumptions of Theorem \ref{thm:MainR} into two subsets. This provides a further understanding of strategically simple mechanisms. If the assumptions of Theorem \ref{thm:MainR} hold, then, for any $R\in \mathcal{R}_1\times \mathcal{R}_2 \times \ldots \times \mathcal{R}_n$, we denote by $I^*(R)$ the set of local dictators at $R$.

\begin{definition}
Suppose for every agent $i$ there is a non-empty set $\mathcal{R}_i\subseteq \mathcal{R}$ such that $\mathbf{U}_i=\bigcup_{R_i \in \mathcal{R}_i}\mathcal{U}(R_i)$, and suppose $\mathbf{M}_i = \Delta(\mathbf{U}_{-i})$ for all $i\in I$. Then
a strategically simple mechanism with respect to $\textbf{U}$ and $\textbf{M}$ is of ``type 1'' if:
\[
\bigcap\limits_{R \in \bigtimes_{i\in I} \mathcal{R}_i}I^*(R)\neq \emptyset.
\]
Otherwise, it is of ``type 2.''
\end{definition}

In words, in a type 1 strategically simple mechanism, there is an agent who is local dictator at \textit{all} preference profiles, whereas this is not the case for type 2 strategically simple mechanisms.\medskip

Type 1 strategically simple mechanisms can be easily characterized. To state this characterization, we first introduce a class of mechanisms that we call ``delegation mechanisms.''

\begin{definition}\label{DM}
A mechanism is a ``delegation mechanism'' if it is the reduced normal form of an extensive form mechanism of the following type: First, some agent $i^*\in I$ chooses an element $s_{i^*}$ from some finite set $S_{i^*}$. All agents observe $s_{i^*}$. Then, for every $s_{i^*}$, a subgame with simultaneous moves follows in which the players are the agents in $I\setminus \{i^*\}$, and in which a dominant strategy mechanism with outcomes in $A$ is played, where the mechanism may depend on $s_{i^*}$.
\end{definition}

In a delegation mechanism, the mechanism designer delegates the choice of the mechanism to some agent $i^*$. This agent has to choose a mechanism from a given set of dominant strategy mechanisms that the mechanism designer has specified. Clearly, in a delegation mechanism, all agents except $i^*$ have dominant strategies, and therefore do not even have to form first-order beliefs, and for agent $i^*$ therefore only first-order belief is relevant to the optimal choice. It is worth noting that dominant strategy mechanisms, in which all agents, for all relevant utility functions, have dominant strategies, are trivially delegation mechanisms.

\begin{theorem} \label{thm:MM}
Suppose for every agent $i$ there is a non-empty set $\mathcal{R}_i\subseteq \mathcal{R}$ such that $\mathbf{U}_i=\bigcup_{R_i \in \mathcal{R}_i}\mathcal{U}(R_i)$, and suppose $\mathbf{M}_i = \Delta(\mathbf{U}_{-i})$ for all $i \in I$. \smallskip
\begin{enumerate}
\item[(1)] Every delegation mechanism is a type 1 strategically simple mechanism. \smallskip
\item[(2)] For every type 1 strategically simple mechanism, there is an equivalent delegation mechanism.
\end{enumerate}
\end{theorem}

We do not have a parallel result for type 2 strategically simple mechanisms in the general framework. But we shall present a characterization of type 2 strategically simple mechanisms in the bilateral trade environment that we study in the next section, and also a characterization of type 2 strategically simple mechanisms in the voting problem in Section \ref{sec:voting} when there are 2 agents and 3 alternatives.

\section{Bilateral Trade} \label{sec:bilateral-trade}

\smallskip

We first consider an example of an environment in which outcomes include money payments, and in which it is therefore natural to restrict attention to preferences that are monotonically increasing in money, and to beliefs that attach probability 1 to preferences that are monotonically increasing in money. The set of agents is: $I=\{S,B\}$, where $S$ is the seller, and $B$ is the buyer. The set of outcomes is: $A=\{\phi \} \cup  T$, where ``$\phi$'' stands for ``no trade,'' and  $T$ is a finite subset of $\mathds{R}_{++}$. An outcome $t\in T$ corresponds to trade at price $t$. We require trade to be voluntary. We refer to any mechanism for this setting as a ``bilateral trade mechanism'' if each agent has a strategy that enforces the no trade outcome.\medskip

The preferences $R_S$ over $A$ that we consider for the seller are indexed by some value $v_S>0$, and the preferences $R_B$ over $A$ that we consider for the buyer are indexed by some value $v_B>0$. We assume that for $i = S, B$ the set of possible values of $v_i$ is a finite subset $V_i$ of $\mathds{R}_{++}$ with the following properties: $\min V_i < \min T$, $\max V_i > \max T$, and $V_i \cap T = \emptyset$ for $i = S, B$. The preference with index $v_S$ is such that the seller prefers outcome $\phi$ to outcome $t$ if and only if $t < v_S$, and the seller prefers larger elements of $T$ to smaller ones. The preference with index $v_B$ is such that the buyer prefers outcome $\phi$ to outcome $t$ if and only if $t > v_B$, and the buyer prefers smaller elements of $T$ to larger ones.\medskip

In the notation of Section \ref{sec:characterization}, we have now specified the sets $\mathcal{R}_i$ for $i = S, B$. The sets of admissible utility functions $\mathbf{U}_i$ and admissible beliefs $\mathbf{M}_i$ are as given in the first sentence of Theorem \ref{thm:MainR}. Note that the model that we have described does not assume quasi-linear preferences. Rather, arbitrary risk attitudes are allowed.\medskip

Theorem \ref{thm:MM} implies the following characterization of type 1 strategically simple bilateral trade mechanisms:

\begin{proposition} \label{BT1}
A bilateral trade mechanism is type 1 strategically simple if and only if it is equivalent to the normalform of a mechanism of the following type: Agents play a two-stage game of perfect information. \smallskip
\begin{enumerate}
\item[1.] Agent $i^*$ either chooses a price $t$ from some finite set $\hat T \subseteq T$, or chooses to reject trade. If agent $i^*$ rejects trade, then the game ends. No trade takes place, and no transfers are paid. Otherwise, Stage 2 is entered. \smallskip
\item[2.] Agent $-i^*$ accepts or rejects trade at the price $t$ proposed by agent $i^*$. If agent $-i^*$ accepts, then trade takes place, and the buyer pays the seller price $t$. Otherwise, no trade takes place, and no transfers are paid.
\end{enumerate}
\end{proposition}

To obtain the class of mechanisms described in Proposition \ref{BT1}, consider the following simple argument. When there are only two agents, the second-stage dominant strategy mechanisms as referred to in Theorem \ref{thm:MM} are single agent mechanisms in which the agent $-i^*$ chooses among alternatives offered by agent $i^*$. Among the options offered that do include trade, the seller, if she is agent $-i^*$, will always pick trade at the highest price, and the buyer, if he is agent $-i^*$, will always pick trade at the lowest price. Therefore, offering trade at more than one price is redundant. Moreover, the mechanism that the seller offers must always include the no trade option.\medskip

Proposition \ref{BT1} in fact provides a complete characterization of all bilateral trade mechanisms that are strategically simple, as the following result, which we prove in Appendix \ref{app:bilateral-trade}, shows:

\begin{proposition} \label{BT2}
There are no bilateral trade mechanisms that are type 2 strategically simple.
\end{proposition}

Strategically simple mechanisms are more flexible than dominant strategy mechanisms in the bilateral trade environment. It is known that the only dominant strategy mechanisms that satisfy ex post budget balance and individual rationality are posted price mechanisms. As discussed in the Introduction, for each posted price mechanism, there exists a corresponding type 1 strategically simple bilateral trade mechanism (the price cap mechanism) that facilitates more efficient trade.

\section{Voting} \label{sec:voting}

\smallskip

We now analyze strategically simple mechanisms in settings in which no restrictions are assumed regarding the agents' utilities or beliefs: $\mathbf{U}_i = \mathcal{U}$ and $\mathbf{M}_i = \Delta(\mathcal{U}^{n-1})$ for all $i \in I$. Note that this is the most demanding form of strategic simplicity. We call a mechanism that is strategically simple on this domain a ``strategically simple voting mechanism,'' because the unrestricted domain is a domain that has been considered in parts of the voting literature. The celebrated Gibbard-Satterthwaite theorem \cite{Gibbard1973, Satterthwaite1975} shows that in the voting environment, a mechanism has dominant strategies if and only if it is dictatorial. As we shall see, there are many more strategically simple voting mechanisms.\medskip

The voting environment satisfies the domain assumptions in Section \ref{sec:characterization}. Thus, Theorems \ref{thm:MainR} and \ref{thm:MM} can be applied. We shall now distinguish type 1 and type 2 strategically simple voting mechanisms. In a type 1 strategically simple voting mechanism, some agent $i^*$ chooses a subset of the set $A$ of alternatives and a dominant strategy mechanism for the other agents to pick one alternative from this set. In a second stage, the other agents then play this dominant strategy mechanism. The influence of the first agent on the ultimate outcome may be restricted by limiting the set of subsets of $A$ and dominant strategy mechanisms she can choose from. If she can choose any arbitrary subset, then, of course, we have the classical dictatorship.\medskip

Standard results in voting theory provide characterizations of dominant strategy mechanisms that can be played in the second stage. If agent $i^*$ rules out all but two alternatives, then a mechanism has dominant strategies if and only if it is a generalized form of majority voting (see Barber\`a \cite[p. 759]{Barbera2010}). If agent $i^*$ allows the other agents to pick from at least three alternatives, then it follows from the Gibbard-Satterthwaite theorem that only dictatorial mechanisms have dominant strategies. Thus, agent $i^*$, if she wants to allow at least three alternatives,  has to pick one of the other agents, and needs to let this agent make the ultimate decision, where this agent is restricted to the set of alternatives chosen by agent $i^*$.\medskip

Type 2 strategically simple voting mechanisms are harder to characterize. Here, we provide a characterization for the voting environment with two agents and three alternatives, but leave as an open question the full characterization for more general voting environments.

\begin{proposition} \label{prop:voting-type-2}

Let $\#A= 3$ and $n=2$. A voting mechanism is type 2 strategically simple if and only if it is one of the following two mechanisms (up to relabeling of the agents and the alternatives):\bigskip

\begin{center}
\begin{tabular}
[|C{0.4cm}]{|C{0.4cm}|C{0.4cm}|C{0.4cm}|C{0.4cm}|C{0.4cm}|C{0.4cm}|} \hline
& $a$ & $b+$ & $b-$ & $c+$ & $c-$\\ \hline
$a$ & $a$ & $a$&$a$&$a$&$a$ \\ \hline
$b+$ & $a$ & $b$&$b$&$a$&$b$ \\ \hline
$b-$ & $a$ & $b$&$b$&$c$&$b$ \\ \hline
$c+$ & $a$ & $a$&$c$&$c$&$c$ \\ \hline
$c-$ & $a$ & $b$&$b$&$c$&$c$ \\ \hline
\end{tabular} \quad \quad \quad
\begin{tabular}
[|C{0.4cm}]{|C{0.4cm}|C{0.4cm}|C{0.4cm}|C{0.4cm}|C{0.4cm}|} \hline
& $a$ & $b$ & $c+$ & $c-$ \\ \hline
$a$ & $a$ & $a$&$a$&$a$ \\ \hline
$b+$ & $a$ & $b$&$a$&$b$ \\ \hline
$b-$ & $a$ & $b$&$c$&$b$ \\ \hline
$c$ & $a$ & $b$&$c$&$c$ \\ \hline
\end{tabular}\bigskip

Mechanism A \hskip 3cm Mechanism B
\end{center}
\end{proposition}\medskip

The second mechanism in Proposition \ref{prop:voting-type-2} has already appeared in Section \ref{sec:examples} as an illustration of the notion of strategic simplicity (although the strategies were labeled differently there).  Here, we give a more extensive discussion of the first mechanism, and then add some brief comments about the second mechanism. For simplicity, let us call the first mechanism ``mechanism A'' and the second mechanism ``mechanism B.''\medskip

\textbf{Interpretation of Mechanism A:} This mechanism has the following interpretation (which also motivates the labels that we have given to the strategies).  Each agent has five strategies: a vote for $a$, a ``strong vote'' for $b$ denoted by $b+$, a ``weak vote'' for $b$ denoted by $b-$, and similarly a ``strong vote'' and a ``weak vote'' for $c$, denoted by $c+$ and $c-$ respectively. Alternative $a$ is the default alternative. If at least one of the agents votes for the default, then the default is chosen. If both agents vote for $b$ (resp. $c$), then $b$ (resp. $c$) is chosen regardless of whether the votes are strong or weak. If one of the agents casts a strong vote for $b$, but the other agent only casts a weak vote for $c$, then $b$ is chosen.  Similarly, if one agent casts a strong vote for $c$, but the other agent only casts a weak vote for $b$, then $c$ is chosen. If the agents cast weak votes for different alternatives, the disagreement is resolved in favor of $b$. If the agents cast strong votes for different alternatives, the mechanism reverts to the default $a$.\medskip

\textbf{Voting Incentives in Mechanism A:} Let us verify that this mechanism is strategically simple. Agents who rank $a$ top have a dominant strategy to vote for $a$.  Let us say that an agent $i$ has ``a weak preference for $b$ over $c$'' if her preference is: $bR_icR_ia$. That is, she ranks the default $a$ below both $b$ and $c$. For such an agent, a weak vote for $b$ is weakly dominant. Similarly, for an agent who has a ``strong preference for $b$ over $c$,'' i.e., $cR_ibR_ia$, a strong vote for $b$ is weakly dominant.\medskip

Weak and strong preferences for $c$ over $b$ are analogously defined, and it is also clear that an agent with a strong preference for $c$ over $b$ has the weakly dominant strategy of voting strongly for $c$. The final case to consider is an agent who has a weak preference for $c$ over $b$. Such an agent has two undominated strategies: a ``weak'' or a ``strong'' vote for $c$. Informally, such an agent potentially has an incentive to ``overstate'' the strength of her preference. The reason is as follows. A weak vote for $c$ will inevitably lead to $b$ if the other agent votes for $b$. If a strong vote for $c$ is cast, then the outcome is $a$, which is worse than $b$,  if the other agent casts a strong vote for $b$, and the outcome is $c$, which is better than $b$, if the other agent casts a weak vote for $b$. Which of these two cases is more important depends on the utility difference between $c$ and $b$, and on the agent's belief about the relative likelihood of the other agent having a weak or strong preference for $b$ over $c$.\medskip

Notice that in no case higher-order beliefs matter for an agent's choice. This is obvious if an agent has a weakly dominant choice. For agents with a weak preference for $c$ over $b$, higher-order beliefs don't matter because, whether the other agent casts a weak or strong vote for $b$ does not depend on that agent's first-order beliefs. Thus, the mechanism is strategically simple.\medskip

\textbf{Normative Properties of Mechanism A:} Why might a mechanism designer find mechanism A attractive? The most obvious attractive feature of mechanism A is that it is anonymous, that is, it treats all agents equally. Anonymity is often regarded by itself as a desirable property of a voting mechanism. No type 1 strategically simple mechanism, and in particular no dominant strategy mechanism, is anonymous in the voting setting with 2 agents and 3 alternatives.\medskip

The mechanism may also appeal to a mechanism designer who maximizes expected welfare. To show this, we consider two welfare criteria: the sum of agents' utilities (``utilitarian welfare'') and the minimum of agents' utilities (``Rawlsian welfare''), and we consider the comparison between mechanism A and dictatorship. Under mechanism A, when both agents rank $a$ in the middle, but rank different outcomes top, outcome $a$ is chosen, and thus  agents ``compromise.'' This yields higher Rawlsian welfare than dictatorship, and it might also yield higher utilitarian welfare than dictatorship, depending on the agents' vNM utilities of the compromise. When one agent ranks $c$ top, and the other agent ranks $c$ second behind $b$, then $c$ is chosen, regardless of which agent ranks $c$ top, whereas under dictatorship $c$ is only chosen when the first agent is the dictator. If the second agent is the dictator, then mechanism A yields higher utilitarian and Rawlsian welfare.\medskip

On the other hand, mechanism A might also lead to a Pareto inefficiency. Pareto inefficiencies harm both utilitarian welfare and Rawlsian welfare, and they are not possible under dictatorship. A Pareto inefficiency occurs if one agent has a weak preference for $c$ over $b$ and the other agent has a strong preference for $b$ over $c$. Agents may end up with $a$, although both prefer $b$ to $a$. This happens if the agent with a weak preference for $c$ over $b$ exaggerates the strength of her preference and casts a strong vote for $c$. But she will do so only if the utility from $b$ is close to zero, and therefore the loss from getting $a$ rather than $b$ is ``small.''\medskip

To compare mechanism A and dictatorship, the designer might adopt an ex ante perspective and calculate the expected welfare if she has a prior over all utility functions and all first-order beliefs. It is clear from the analysis in the previous two paragraphs that depending on which prior the designer uses, she may prefer one mechanism or the other. Under the uniform distribution over all utility functions and all first-order beliefs,  mechanism A achieves higher expected welfare than the dictatorial mechanism under both welfare criteria, as one can check numerically.\medskip

\textbf{Mechanism B:} This mechanism is similar to mechanism A, but it is not anonymous.  In mechanism B, only agent 1 can differentiate between a weak or a strong vote for $b$, and only agent 2 can differentiate between a weak or a strong vote for $c$. The voting rules are then similar to the voting rules  in mechanism A. In this mechanism, as the analysis in Section \ref{sec:examples} showed, agent 1 has two undominated strategies when she has a ``strong'' preference for $c$ over $b$, i.e. $cR_iaR_ib$: voting for $c$ and voting for $a$. This reflects that she cannot cast a ``strong'' vote for $c$, unlike in mechanism A. Agent 2 has two undominated strategies when she has a weak preference for $c$ over $b$. As in mechanism A, she might cast a ``weak'' or a ``strong'' vote for $c$ in this case.

\begin{remark}
So far in the voting problem, we have considered the domain in which there is no restriction regarding the agents' utilities or beliefs: $\mathbf{U}_i = \mathcal{U}$ and $\mathbf{M}_i = \Delta(\mathcal{U}^{n-1})$ for all $i \in I$. A common domain restriction for preferences in the voting literature is the so-called single-peaked domain. In our context, a single-peaked domain would be one for which there is an ordering of the alternatives in $A$ so that for all $i$ all utility functions in $\mathbf{U}_i$ are single-peaked, and the set $\mathbf{M}_i$ is the set of all beliefs that assign probability 1 to single-peaked utility functions.  It is easy to see that the first mechanism in Proposition \ref{prop:voting-type-2} is type 2 strategically simple even in the single-peaked domain when the alternatives are arranged in alphabetical order, whereas the second mechanism is type 1 strategically simple on the single-peaked domain.
\end{remark}

\section{Related Literature} \label{sec:literature}

\smallskip

Li \cite{Li2017} proposes the concept of ``obviously strategy-proof mechanisms.'' These mechanisms form a subclass of dominant strategy mechanisms in which it is particularly easy for agents to recognize that they have a dominant strategy. While Li's work is, in spirit, related to ours, our purpose is to introduce a class of mechanisms that is larger (rather than smaller) than the class of dominant strategy mechanisms, yet consists only of ``simple'' mechanisms. Our motivation for this is that in many applications the set of dominant strategy mechanisms seems ``too small.''\medskip

Li starts with the observation that subjects in experiments often do not recognize dominant strategies, but that they do recognize such strategies if the mechanism is ``obviously strategy-proof.''  But if subjects in experiments don't even recognize what is not obvious, readers might ask, then how can we expect them to engage in the strategic reasoning that we have called ``strategically simple'' in this paper?\medskip

We think that Li captures a different dimension of ``simplicity'' than we do. In dominant strategy mechanisms, and arguably also in strategically simple mechanisms, the ``logic'' that underlies the determination of an optimal strategy is straightforward. By this we mean that the logic can be explained to the agents in a simple and persuasive way. This is presumably a necessary, but not sufficient condition for optimal choices to be ``obvious'' in the sense that agents can easily find optimal strategies by themselves, without being offered explanations. In practice, it seems common that mechanism designers spend a lot of time explaining to the participants in the mechanism how the mechanism works, and which considerations the participants should base their strategic choices on. We regard our requirement of strategic simplicity as a formalization of the idea that the mechanism designer can present a simple and persuasive explanation of the relevant strategic considerations to the agents. That does not mean that the optimal choices are ``obvious'' to the agents.\medskip

Several recent papers have analyzed mechanism design when agents' strategy choices are guided by ``level-$k$ thinking.''\footnote{De Clippel et al. \cite{DeClippel2017}, Crawford \cite{Crawford2016}, Gorelkina \cite{Gorelkina2018}, Kneeland \cite{Kneeland2017}.}  Like our paper, the level-$k$ model of strategic choice is motivated by the idea that there is bound to the length of hierarchies of beliefs that agents can form. We assume in this paper that $k$ is equal to 2. The level-$k$ model relies, however, on an exogenously assumed ``anchor'' that describes the beliefs of an agent who does not analyze the other agents' incentives at all, but who does maximize expected utility (``level 1 agents'').  This amounts to selecting among the not strictly dominated strategies of level 1 agents those that are best responses to a particular belief. We select the not weakly dominated strategies, and thus implicitly assume full support beliefs, but do not fix any specific anchor belief. Thus, our theory of behavior is in most games more permissive than the level-$k$ model of behavior.\medskip 

Particularly closely related to our work is a paper by de Clippel et al. \cite{DeClippel2017}. They consider an incomplete information environment with a  common prior type space. The mechanism designer seeks to implement a social choice function that assigns to each type vector one outcome. They study a mechanism designer who believes that each agent $i$ is a level $k_i(\geq 1)$ player, but who does not know the levels $k_i$. Implementation is achieved if for each type vector the desired outcome results whenever each player $i$ is a level $k_i$ player for any combination of $k_i$'s. This implies that players with a level $k_i\geq 2$ anticipate the same outcome regardless of which level the other players are, and therefore that there is no benefit to a player of thinking beyond level 2, just as in our paper. However, unlike our paper, de Clippel et al.'s postulate an exogenous anchor, work with common priors, and focus on social choice functions rather than correspondences.

Our work is also related to papers that consider the implementation of social choice functions when agents perform a limited number of rounds of elimination of dominated strategies. Saran's \cite{Saran2016} implementation notion includes the requirement that any strategy combination that survives one round of elimination of strictly undominated strategies yields the outcome prescribed by the social choice function. He obtains for many economic environments that a strict subset of the set of all strategy-proof social choice functions can be implemented. By comparison, in our model agents are assumed to be able to perform higher-order strategic thinking. Another difference is that we allow for the elimination of weakly dominated strategies.\medskip

Jackson et al. \cite{Jackson1994} and Sj{\"o}str{\"o}m \cite{Sjostrom1994} show that, in certain economic environments, all social choice functions can be implemented in two rounds of elimination of weakly dominated strategies, with a unique strategy surviving in round two. Thus, the corresponding mechanisms are strategically simple on the restricted domain in which agents have point beliefs about the other agents' preferences. This paper considers strategic simplicity on larger domains.\medskip

In complete information models, de Clippel et al. \cite{DeClippel2014} and Van der Linden \cite{VanDerLinden2017} use the number of rounds of elimination of dominated strategies, or of backward induction, that are required to solve a mechanism as a measure of the strategic complexity of mechanisms for the choice of an arbitrator or of a jury. This idea is closely related to our notion of strategic simplicity. One important difference with our work is that they don't allow uncertainty about other agents' preferences.\medskip

Bahel and Sprumont \cite{BahelSprumont2017} consider dominant strategy mechanisms for the choice among Savage acts. The act that is chosen by the mechanism may depend on each agents' beliefs about the state, but it will not depend on any agent's beliefs about the other agents' beliefs about the state, etc. This is because, for given beliefs and valuations, their mechanisms have dominant strategies. There is thus a parallel between their work and ours, although in their work beliefs are about Savage-style ``states of the world,'' whereas in our work beliefs are about other agents' preferences.\medskip

Strategic simplicity can also be interpreted as a form of robustness in the sense of Bergemann and Morris \cite{BergemannMorris2005}. Whereas Bergemann and Morris study implementation that does not rely on any conditions on agents' hierarchies of beliefs, we study implementation of outcomes that may depend on agents' first layer of beliefs, but not on any higher-order beliefs.

\section{Discussion} \label{sec:discussion}

\smallskip

In this section, we suggest some directions for further research. First, it might be interesting to define a refinement of strategic simplicity that is satisfied by all mechanisms that we call ``type 1 strategically simple,'' but not by any mechanism that we call ``type 2 strategically simple.'' Type 2 strategically simple mechanisms don't just seem harder to characterize. The type 2 strategically simple mechanisms that we have found in the voting environment also seem more complicated for the agents than the delegation mechanisms.\footnote{We thank the editor and the referees for encouraging us to think along this direction.}\medskip


To find such a refinement, one would have to look for an intermediate notion of simplicity that is weaker than the requirement that agents have dominant strategies but stronger than the requirement that the mechanism be strategically simple in our sense. One way of proceeding would be the following. Say that a mechanism is ``strategically simple$^*$'' if each agent believes their opponents choose a dominant strategy if there is one, but she does not have any idea how the other agents choose their strategies if there is no dominant strategy. In particular, she may even consider it plausible that they choose dominated strategies if there is no dominant strategy. Note that this alternative notion lowers each agent's belief in the other agents' rationality in comparison to our construction. This alternative assumption enlarges the set of compatible strategic beliefs and thus makes it harder for each agent to have an action that is always a best response to all these compatible strategic beliefs. This does not affect the type 1 strategically simple mechanisms that we study, because the agents who have dominant strategies still do not have to think about the other agents' beliefs, and the only agent who does not have dominant strategies believes that the other agents play the dominant strategy. However, this would rule out type 2 strategically simple mechanisms as ``strategically simple$^*$'' mechanisms, because in type 2 strategically simple mechanisms, there are at least two agents who do not have a dominant strategy, and each of these two agents does not believe the other agent only picks undominated strategies.\medskip

The construction described in the previous paragraph is, however, a bit ad-hoc. Thus the open research question is whether there are other ways of strengthening the requirement of strategic simplicity that rule out type 2 strategically simple mechanisms, but not type 1 strategically simple mechanisms. Another interesting research direction on type 2 strategically simple mechanisms would be to seek general conditions on primitives that imply that all strategically simple mechanisms are of type 1. We provided one such environment in this paper, namely, the bilateral trade environment. Our proof in the bilateral trade environment relies heavily on the particular domain structure, and we have not yet found useful ways to generalize the result to other settings.\medskip

Strategic simplicity as defined in this paper focuses on mechanisms in which the agents' optimal choices can be based on first-order beliefs alone. It would be interesting to investigate mechanisms in which there is some integer $k\geq 2$ such that only beliefs up to the \textit{k}th-order matter for agents' choices. Our experience from writing this paper suggests that a characterization of such mechanisms might be particularly difficult when considering the analog of what we have called in this paper ``type 2'' strategically simple mechanisms. However, even the generalization of type 1 strategically mechanism seems non-trivial. Consider the following simple example: There are three agents that vote over four alternatives. The three agents in turn remove one alternative from consideration, and the remaining alternative is chosen. This is a simple extension of type 1 strategically simple voting mechanisms, and one might conjecture that the highest order beliefs that matter for optimal choice in this mechanism are the second-order beliefs. Obviously, the agent who moves last does not have to form any beliefs. One might conjecture that the agent who moves second only has to form first-order belief, and that the agent who moves first only has to form second-order belief. However, this conjecture is wrong. The reason is that the second mover needs to form second-order belief about the correlation between the first mover's belief and the third mover's preference, and therefore, for the first mover, it matters what he believes about the second-order belief of the second mover.  We have not yet tackled the complex issues raised by this example, but hope to do so in future research.\medskip

A study of strategically simple mechanisms in applications other than the two settings covered in this paper would be interesting. For certain classes of environments with quasi-linear preferences, mechanisms in which agents need to form at most first-order beliefs to determine their expected utility maximizing strategies have been described in Chen and Li \cite{ChenLi2018}, Yamashita and Zhu \cite{YamashitaZhu2017}, and Cr\'emer and Riordan \cite{CremerRiordan1985}. The first two papers show that such strategically simple mechanisms dominate the optimal dominant strategy mechanism for a revenue maximizing mechanism designer, and Cr\'emer and Riordan \cite{CremerRiordan1985} focus on efficiency properties. While strategic simplicity is not the focus of these papers, their results suggest that a further study of strategically simple mechanisms in environments with quasi-linear preferences (for example, strategically simple auctions) might be promising.\medskip


Future research might also develop criteria which a mechanism designer could use to evaluate strategically simple mechanisms, and then characterize using such criteria the best strategically simple mechanisms. The simplest way of proceeding would be to endow the mechanism designer with a prior over agents' utility functions and their first-order beliefs, and then to maximize the expected value of the designer's objective function. In this paper we have conducted such an exercise when comparing the expected welfare from a type 2 strategically simple voting mechanism and a dictatorial voting mechanism. The comparison is based on a uniform prior. An interesting open research question is whether in the voting context, and in general, results can be obtained that do not rely on the ad-hoc specification of a prior.\medskip

Finally, experimental tests of our notion of strategic simplicity would be of interest. We argued in the previous section that in strategically simple mechanisms the ``logic'' that underlies the determination of an optimal strategy is straightforward, although it might not be obvious. To address the possibility that experimental subjects might not discover by themselves what is not obvious, the experimenter might explain it to them. A test of our concept of strategic simplicity could potentially rely on a comparison between the subjects' understanding and acceptance of the experimenter's explanations in strategically simple mechanisms, and in mechanisms that are not strategically simple.



\newpage

\appendix

\section{Proof of Theorem 1}

\smallskip

We first show that mechanisms that satisfy the conditions in the theorem are strategically simple.  Fix an agent $i \in I$, any utility function $u_i \in \mathbf{U}_i$, and any utility belief $\mu_i \in \mathbf{M}_i$. We have to show that agent $i$ has a strategy that is a best response to all compatible strategic beliefs.\medskip

Let $R_i$ denote the preference induced by $u_i$. It suffices to show that there is a strategy in $U\!D_i(R_i)$ that is among the strategies in $U\!D_i(R_i)$ a best response to all compatible strategic beliefs. By the definition of weak dominance given $R_i$,  the same strategy will then also be among \emph{all} strategies of agent $i$ a best response to all compatible strategic beliefs.\medskip

 We can classify the profiles of utility functions $u_{-i}$ of agents other than $i$  into two categories: (1) the ones that induce preference profile $R_{-i}$ such that agent $i$ is a local dictator in the mechanism restricted to the strategy set $U\!D_i(R_i)$ for agent $i$ and strategy sets $U\!D_{-i} (R_{-i})$ for the other agents; and (2) the ones that induce preference profile $R_{-i}$ such that agent $i$ is not a local dictator in this mechanism. In the first case, the outcome is determined by agent $i$'s own choice from $U\!D_i(R_i)$ regardless of which strategies in $U\!D(R_{-i})$ the other agents choose, and with utility function $u_{-i}$ all undominated strategies of the other agents will be contained in $U\!D_{-i}(R_{-i})$.   In the second case, the outcome is  the same regardless of agent $i$'s own choice from $U\!D_i(R_i)$ as long as the other agents choose strategies in $U\!D_{-i}(R_{-i})$, and, again, with utility function $u_{-i}$ all undominated strategies of the other agents will be contained in $U\!D_{-i}(R_{-i})$. Thus, agent $i$'s expected utility maximizing choice from $U\!D_i(R_i)$ only depends on her utility belief, and is the same for all compatible strategic beliefs, and we can conclude that the mechanism is strategically simple.\medskip

Next, we show that mechanisms that are strategically simple with respect to domains described in the theorem must satisfy the conditions in the theorem. We proceed by establishing a sequence of claims.

\begin{claim}\label{Indiff1}
	Let $u_i \in \mathbf{U}_i$, $u_{-i} \in \mathbf{U}_{-i}$, and let $\mu_i \in \mathbf{M}_i$ be a utility belief such that $\mu_i(\{u_{-i}\})>0$. Suppose $s_i,s_i'\in \bigcap_{\hat \mu_i \in \mathcal{M}_i(\mu_i)} B\!R_i(u_i,\hat \mu_i)$. Then for all $s_{-i},s_{-i}' \in U\!D_{-i}(u_{-i})$:
	\[
	u_i(g(s_i,s_{-i}))-u_i(g(s_i',s_{-i})) = u_i(g(s_i,s_{-i}'))-u_i(g(s_i',s_{-i}')).
	\]
\end{claim}

\begin{proof}
	Suppose the assertion were not true. Then there are  $s_{-i}, s_{-i}'\in U\!D_{-i}(u_{-i})$ such that:
	\begin{align*}
		u_i(g(s_i,s_{-i}))-u_i(g(s_i',s_{-i})) > u_i(g(s_i, s_{-i}'))-u_i(g(s_i', s_{-i}')).
	\end{align*}
	Pick any $\hat \mu_i \in \mathcal{M}_i(\mu_i)$ that places strictly positive probability on $s_{-i}$ and $s_{-i}'$. Because  $s_i$ and $s_i'$ are both in $B\!R_i(u_i,\hat \mu_i)$ both strategies must yield the same expected utility under $\hat \mu_i$. Now suppose we vary $\hat \mu_i$ such that it places $\varepsilon$ probability more than $\hat \mu_i$ on $s_{-i}$ and $\varepsilon$ probability less than $\hat \mu_i$ on $s_{-i}'$, leaving all other probabilities unchanged. If we choose $\varepsilon>0$ and sufficiently small, we can vary $\hat \mu_i$ in this way so that it remains an element of $\mathcal{M}_i(\mu_i)$, and so that for the modified belief $s_i$ is a strictly better response than $s_i'$. This contradicts $s_i'\in \bigcap_{\hat \mu_i \in \mathcal{M}_i(\mu_i)} B\!R_i(u_i,\hat \mu_i)$.
\end{proof}

\begin{claim}\label{Indiff2}
	Let $u_i \in \mathbf{U}_i$, $u_{-i} \in \mathbf{U}_{-i}$, and let $\mu_i, \mu_i' \in \mathbf{M}_i$ be any two utility beliefs such that $\mu_i(\{u_{-i}\})>0$ and $\mu_i'(\{u_{-i}\})>0$. Suppose:
	\begin{align*}
		s_i & \in \bigcap_{\hat \mu_i \in \mathcal{M}_i(\mu_i)} B\!R_i(u_i,\hat \mu_i); \\
		\text{and } s_i' & \in \bigcap_{\hat \mu_i' \in \mathcal{M}_i(\mu_i')} B\!R_i(u_i,\hat \mu_i').
	\end{align*}  Then for all $s_{-i},s_{-i}' \in U\!D_{-i}(u_{-i})$:
	\[
	u_i(g(s_i,s_{-i}))-u_i(g(s_i',s_{-i})) = u_i(g(s_i,s_{-i}'))-u_i(g(s_i',s_{-i}')).
	\]
\end{claim}

\begin{proof}
	We focus on the non-trivial case: $s_i\neq s_i'$. Claim \ref{Indiff2} follows from repeated applications of Claim \ref{Indiff1} if we can find a sequence of utility beliefs of agent $i$, $\mu_i^k$ ($k=2,\ldots,K$), and strategies of agent $i$, $s_i^k$ ($k=1,2,\ldots,K$), where $K\geq 2$, such that $s_i^1=s_i$, $s_i^K=s_i'$, for every $k\in \{2,\ldots, K\}$ the utility belief $\mu_i^k$ places positive probability on $u_{-i}$, and for every $k \in \{2,\ldots, K\}$  both $s_i^{k-1}$ and $s_i^k$ are elements of $\bigcap_{\hat \mu_i^k \in \mathcal{M}_i(\mu_i^k)} B\!R_i(u_i,\hat \mu_i^k)$. We shall construct such a sequence.\medskip
	
	For every $\alpha \in [0,1]$ we define $\mu_i(\alpha)\equiv (1-\alpha) \mu_i + \alpha \mu_i'$. We set $s_i^1=s_i$.  Define $\alpha^2\equiv \sup\{\alpha \in [0,1]| s_i^1 \in \bigcap_{\hat \mu_i \in \mathcal{M}_i(\mu_i(\alpha))} B\!R_i(u_i,\hat \mu_i)\}$. Observe that the upper hemi-continuity of the best response correspondence implies that $s_i^1\in  \bigcap_{\hat \mu_i \in \mathcal{M}_i(\mu_i(\alpha^2))} B\!R_i(u_i,\hat \mu_i)\}$. If $\alpha^2=1$, then we can set $s_i^2=s_i'$, $\mu_i^2=\mu_i'$, $K=2$, and our sequence has all the required properties.\medskip
	
	If  $\alpha^2<1$, define $s_i^2$ to be any strategy in $S_i$ that is an element of: $$\bigcap_{\hat \mu_i \in \mathcal{M}_i(\mu_i(\alpha^2+\varepsilon))} B\!R_i(u_i,\hat \mu_i)\}$$ for a sequence of $\varepsilon>0$ tending to zero. Then, by upper hemi-continuity of the correspondence of best responses,  $s_i^2 \in \bigcap_{\hat \mu_i \in \mathcal{M}_i(\mu_i(\alpha^2))} B\!R_i(u_i,\hat \mu_i)\}$. We define $\mu_i^2$ to be $\mu_i(\alpha^2)$. Note that, because $\mu_i$ and $\mu_i'$ attach strictly positive probability to $u_{-i}$, and because $\mu_i^2$ is a convex combination of $\mu_i$ and $\mu_i'$, also $\mu_i^2$ places strictly positive probability on $u_{-i}$. If $s_i^2=s_i'$, then we set $K=2$, and the construction is complete.\medskip
	
	If $s_i^2\neq s_i'$, then we repeat the steps just described. In general, let $k\geq 2$, and suppose that, after $k-1$ steps, we had  determined $\mu_i^k$ such that $\mu_i^k=\mu_i(\alpha^k)$ for some $\alpha^k<1$, and $s_i^k$ such that $s_i^k\neq s_i'$. Then repeating the steps described above means that  we define $\alpha^{k+1}\equiv \sup\{\alpha \in [\alpha^k,1]| s_i^k \in \bigcap_{\hat \mu_i \in \mathcal{M}_i(\mu_i(\alpha))} B\!R_i(u_i,\hat \mu_i)\}$.  By the upper hemi-continuity of the best response correspondence: $s_i^k\in  \bigcap_{\hat \mu_i \in \mathcal{M}_i(\mu_i(\alpha^{k+1}))} B\!R_i(u_i,\hat \mu_i)\}$.  If $\alpha^{k+1}=1$, then we can define $s_i^{k+1}=s_i'$, $\mu_i^{k+1}=\mu_i'$, $K=k+1$, and our sequence has the required properties. If $\alpha^{k+1}<1$, define $s_i^{k+1}$ to be a strategy in $S_i$ that is an element of $\bigcap_{\hat \mu_i \in \mathcal{M}_i(\mu_i(\alpha^{k+1}+\varepsilon))} B\!R_i(u_i,\hat \mu_i)\}$ for a sequence of $\varepsilon>0$ tending to zero. By the upper hemi-continuity of the correspondence of best responses, $s_i^{k+1} \in \bigcap_{\hat \mu_i \in \mathcal{M}_i(\mu_i(\alpha^{k+1}))} B\!R_i(u_i,\hat \mu_i)\}$. We define $\mu_i^{k+1}$ to be $\mu_i(\alpha^{k+1})$. Note that  $\mu_i^{k+1}$ places strictly positive probability on $u_{-i}$. If $s_i^{k+1}=s_i'$, then we set $K=k+1$, and the construction is complete. Otherwise, we continue as before.\medskip
	
	Note that by construction, in the sequence of strategies no strategy is ever repeated.  Because the number of strategies is finite, the construction has to end after a finite number of steps. At that point our sequence will have all the required properties.
\end{proof}

\begin{claim}\label{Concav}
	For every agent $i$, for every preference $R_i \in \mathcal{R}_i$ on $A$, there exists a utility function $u_i^*$ that represents $R_i$, such that for every $s_i\in U\!D_i(R_i)$ there is a strategic belief $\hat \mu_i$ with support equal to $S_{-i}$ such that:
	\[
	B\!R_i(u_i^*,\hat \mu_i) = \{s_i\}.
	\]
	Moreover, the utility function $u_i^*$ can be chosen such that $u_i^*(a)-u_i^*(b)\neq u_i^*(c)-u_i^*(d)$ for all $(a,b), (c,d)\in A^2$ with $(a,b)\neq (c,d)$.
\end{claim}

\begin{proof}
	First note that, if we can find a utility function $u_i^*$ with the property in the first sentence of Claim \ref{Concav}, then we can slightly perturb this utility function so that the property in the first sentence is maintained, but also the condition in the second sentence of Claim \ref{Concav} holds. Therefore, it is sufficient to prove only the first sentence of Claim \ref{Concav}.\medskip
	
	By the Lemma, and the remark in the first paragraph of the proof of that Lemma, in B\"orgers \cite{Borgers1993}, for every strategy $s_i\in U\! D_i(R_i)$ there exist a utility function $u_{s_i}$ that represents $R_i$, and a full support strategic belief $\hat \mu_i$, such that $s_i$ is the unique maximizer of expected utility given that belief. To prove Claim \ref{Concav} it therefore only remains to be shown that the utility functions $u_{s_i}$ can be chosen to be the same for all strategies $s_i\in U\! D_i(R_i)$.\medskip
	
	We begin with the following observation: Suppose that $s_i$ is the unique maximizer of expected utility in $S_i$ for utility function $u_i$ and full support strategic belief $\hat \mu_i$, and suppose that $f: \mathds{R}\rightarrow \mathds{R}$ is strictly increasing and concave. We claim that then there is another full support strategic belief $\hat{\hat \mu}_i$ such that $s_i$ is the unique maximizer of expected utility for the utility function $f\circ u_i$. To see this note first that, because $s_i$ maximizes expected utility for a full support belief if utility is $u_i$, it is not weakly dominated given utility function $u_i$. Next, because $f$ is increasing and concave, $s_i$ is not weakly dominated given utility function $f \circ u_i$, either.  This follows directly from the argument in the proof of Proposition 1 in Weinstein \cite{Weinstein2016}. We can now use Lemma 4 in Pearce \cite{Pearce1984} and conclude that there is some full support strategic belief $\hat{\hat \mu}_i$ of agent $i$ such that $s_i$ maximizes expected utility when the utility function is $f \circ u_i$. It remains to be shown that this belief can be chosen such that $s_i$ is the \emph{unique} maximizer of expected utility. We do this in the next paragraph.\medskip
	
	Because $s_i$ is the unique maximizer of expected utility for some full support belief if the utility function is $u_i$, by Theorem 2.3 in Bertsimas and Tsitsiklis \cite{BertsimasTsitsiklis1997}, the utility vector $\left(u_i(s_i,s_{-i})\right)_{s_{-i}\in S_{-i}}\in \mathds{R}^{|S_{-i}|}$ is an extreme point of the convex hull of the set of all such utility vectors:
	\[
	co\left(\left\{\left(u_i(s_i',s_{-i})\right)_{s_{-i}\in S_{-i}}|s_i'\in S_i\right\}\right).
	\]
	We now claim that the utility vector corresponding to $s_i$ remains an extreme point if we apply an increasing and concave transformation to $u_i$. That is,  we claim that $\left(f(u_i(s_i,s_{-i}))\right)_{s_{-i}\in S_{-i}}\in \mathds{R}^{|S_{-i}|}$ is an extreme point of:
	\[
	co\left(\left\{\left(f(u_i(s_i',s_{-i}))\right)_{s_{-i}\in S_{-i}}|s_i'\in S_i\right\}\right).
	\]
	Suppose it were not. Then $\left(f(u_i(s_i,s_{-i}))\right)_{s_{-i}\in S_{-i}}$ could be written as a convex combination of the elements of $\left\{\left(f(u_i(s_i',s_{-i}))\right)_{s_{-i}\in S_{-i}}|s_i'\in S_i, s_i'\neq s_i\right\}$, that is, there would be a mixed strategy $\sigma_i\in \Delta(S_i)$ of agent $i$ that attaches zero probability to $s_i$, and such that:
	\[\left(f(u_i(s_i,s_{-i}))\right)_{s_{-i}\in S_{-i}}=\sum_{s_i'\in S_i}\left(f (u_i(s_i',s_{-i}))\right)_{s_{-i}\in S_{-i}}\sigma_i(s_i').\]
	Because $f$ is strictly concave, this implies:
	\[\left(u_i(s_i,s_{-i})\right)_{s_{-i}\in S_{-i}} \lvertneqq \left(u_i(\sigma_i,s_{-i})\right)_{s_{-i}\in S_{-i}},\]
	which contradicts that $s_i$ is not weakly dominated for utility function $u_i$. We conclude that $\left(f(u_i(s_i,s_{-i}))\right)_{s_{-i}\in S_{-i}}\in \mathds{R}^{|S_{-i}|}$  is an extreme point. Using again Theorem 2.3 in Bertsimas and Tsitsiklis \cite{BertsimasTsitsiklis1997} this implies that  there is some function $\xi: S_{-i}\rightarrow \mathds{R}$ such that $s_i$ is the unique maximizer of $\sum_{s_{-i}\in S_{-i}}\xi(s_{-i}) f(u_i(s_i,s_{-i}))$ in $S_i$. Let us treat $\xi$ as a vector in $\mathds{R}^{|S_{-i}|}$. One can verify that there must be a small ball around $\xi$ such that for every  vector $\tilde \xi$ in this ball $s_i$ is the unique maximizer of $\sum_{s_{-i}\in S_{-i}}\tilde \xi(s_{-i}) f(u_i(s_i,s_{-i}))$. We can pick from this ball some $\tilde \xi$ such that $\sum_{s_{-i}\in S_{-i}}\tilde \xi(s_{-i}) \neq 0$. Now  consider the vector $\tilde \mu_i$ defined by:
	\[
	\tilde \mu_i(s_{-i})\equiv \frac{\hat{\hat{\mu}}_i+\varepsilon \frac{\tilde \xi(s_{-i})}{\sum_{s_{-i}'\in S_{-i}}\tilde \xi(s_{-i}')}}{1+\varepsilon}
	\]
	for all $s_{-i}\in S_{-i}$. For sufficiently small $\varepsilon>0$ this is a strategic belief.  It is a convex combination of $\hat {\hat \mu}_i$, for which $s_i$ is \emph{a} expected utility maximizer, and of $\tilde \xi$, for which $s_i$ is the \emph{unique}  maximizer of $\sum_{s_{-i}\in S_{-i}}\tilde \xi(s_{-i}) f(u_i(s_i,s_{-i}))$ in $S_i$. Therefore, $s_i$ is the unique expected utility maximizer for the strategic belief  $\tilde \mu_i$.\medskip
	
	We can now complete the proof by showing that there are a utility function $u_i^*$  and, for every $s_i\in U\! D_i(R_i)$, a concave function $f_{s_i}:\mathds{R}\rightarrow \mathds{R}$, such that $u_i^*=f_{s_i}(u_{s_i})$ for all $s_i\in U\!D_i(R_i)$.
	We first construct $u_i^*$. Enumerate the elements of $A$ as $a_1,a_2,\ldots, a_L$ such that $a_L R_ia_{L-1}R_ia_{L-2} R_i \ldots R_i a_1.$ We pick $u_i^*$ to satisfy the following, where the first two lines are a normalization:
	\begin{eqnarray*}
		u_i^*(a_1)&=&0\\
		u_i^*(a_2)&=&1\\
		&\ldots&
	\end{eqnarray*}
	\[u_i^*(a_{\ell-1})<u_i^*(a_\ell)< u_i^*(a_{\ell-1})+ \ldots \]
	\[ \hskip 3cm \ldots \left(u_i^*(a_{\ell-1})-u_i^*(a_{\ell-2})\right)\min\limits_{s_i\in U\!D_i(R_i)} \frac{u_{s_i}(a_\ell)-u_{s_{i}}(a_{\ell-1})}{u_{s_i}(a_{\ell-1})-u_{s_{i}}(a_{\ell-2})}.\]\vskip 0.5cm

\noindent Note that the right most term in the inequality is strictly larger than the left term, so that $u_i^*$ can be constructed, and will be monotonically increasing, and thus compatible with $R_i$.\medskip
	
	We now turn to the construction of the functions $f_{s_i}$. For every $s_i$ we set $f_{s_i}(u_{s_i}(a_\ell))=u_i^*(a_\ell)$ for all $\ell=1,2,\ldots,L$. This defines $f_{s_i}$ for a finite number of elements of $\mathds{R}$ only. However, it is clear that we can extend $f_{s_i}$ to a concave piecewise linear function on $\mathds{R}$ if it satisfies the following concavity condition for the points in which it is defined:
	\[
	\frac{f_{s_i}(u_{s_i}(a_\ell))-f_{s_i}(u_{s_i}(a_{\ell-1}))}{u_{s_i}(a_\ell)-u_{s_{i}}(a_{\ell-1})} \leq \frac{f_{s_i}(u_{s_i}(a_{\ell-1}))-f_{s_i}(u_{s_i}(a_{\ell-2}))}{u_{s_i}(a_{\ell-1})-u_{s_{i}}(a_{\ell-2})}
	\]
	for all $\ell\geq 2$. By the definition of $f_{s_i}$, this inequality is equivalent to:
	\[\frac{u_i^*(a_\ell)-u_i^*(a_{\ell-1})}{u_{s_i}(a_\ell)-u_{s_{i}}(a_{\ell-1})} \leq \frac{u_i^*(a_{\ell-1})-u_i^*(a_{\ell-2})}{u_{s_i}(a_{\ell-1})-u_{s_{i}}(a_{\ell-2})} \Leftrightarrow \]
	\[ u_i^*(a_\ell) \leq u_i^*(a_{\ell-1}) + \ldots \]
	\[ \ldots \left(u_i^*(a_{\ell-1})-u_i^*(a_{\ell-2})\right)  \frac{u_{s_i}(a_\ell)-u_{s_{i}}(a_{\ell-1})}{u_{s_i}(a_{\ell-1})-u_{s_{i}}(a_{\ell-2})} \]\medskip
	which holds by construction.
\end{proof}

\begin{claim}\label{GroupLD1}
	For every agent $i$, for every preference $R_i \in \mathcal{R}_i$ on $A$, and for every $u_{-i} \in \mathbf{U}_{-i}$ either
	
	(i) there is for every strategy $s_i\in U\!D(R_i)$ an alternative $a$ such that $g(s_i,s_{-i})=a$ for all $s_{-i}\in U\!D_{-i}(u_{-i})$,
	
	or
	
	(ii) there is for every strategy combination $s_{-i}\in U\!D_{-i}(u_{-i})$ an alternative $a$ such that $g(s_i,s_{-i})=a$ for all $s_i\in U\!D_i(R_i)$,
	
	or both.
\end{claim}

\begin{proof}
	Let us represent $R_i$ by the utility function $u_i^*$ from Claim \ref{Concav}. Pick any two $s_i,s_i'\in U\!D_i(R_i)$. By Claim \ref{Concav} there are a full support strategic belief $\hat \mu_i$  such that: $B\!R_i(u_i^*,\hat \mu_i) = \{s_i\}$, and a full support strategic belief $\hat \mu_i'$ such that: $B\!R_i(u_i^*,\hat \mu_i') = \{s_i\}$. Because $\hat \mu_i$ has full support, and because every strategy is undominated for at least some utility function, there is a utility belief $\mu_i$ with $\mu_i(u_{-i})>0$ that is compatible with $\hat \mu_i$. Similarly, there is a utility belief $\mu_i'$ with $\mu_i'(u_{-i})>0$ that is compatible with $\hat \mu_i'$. This implies $s_i\in \bigcap_{\hat \mu_i \in \mathcal{M}_i(\mu_i)} B\!R_i(u_i,\hat \mu_i)$ and $s_i'\in \bigcap_{\hat \mu_i' \in \mathcal{M}_i(\mu_i')} B\!R_i(u_i,\hat \mu_i')$. Therefore, by Claim \ref{Indiff2} for all $s_{-i},s_{-i}' \in U\!D_{-i}(u_{-i})$:
	\[
	u_i^*(g(s_i,s_{-i}))-u_i^*(g(s_i',s_{-i})) = u_i^*(g(s_i,s_{-i}'))-u_i^*(g(s_i',s_{-i}')). \hskip 1cm (*)
	\]
	This has to hold for any two $s_i,s_i'\in U\!D_i(R_i)$.\medskip
	
	Now let us fix some $s_i\in U\!D_i(R_i)$, and suppose first that for some $a\in A$ we have: $g(s_i,s_{-i})=a$ for all $s_{-i}\in U\!D_{-i}(u_{-i})$. Then (*) implies that for every other $s_i'\in U\!D_i(R_i)$ there must be some $\tilde a \in A$ such that $g(s_i,s_{-i})=\tilde a$ for all $s_{-i}\in U\!D_{-i}(u_{-i})$. This follows from $u_i^*(a)-u_i^*(b)\neq u_i^*(c)-u_i^*(d)$ for all $(a,b), (c,d)\in A^2$ with $(a,b)\neq (c,d)$. Thus, we have obtained Case (i).\medskip
	
	Next suppose that for the $s_i$ that we fixed in the previous paragraph we have: $g(s_i,s_{-i})\neq g(s_i,s_{-i}')$ for some $s_{-i},s_{-i}' \in U\!D_{-i}(u_{-i})$. Then  $u_i^*(a)-u_i^*(b)\neq u_i^*(c)-u_i^*(d)$ for all $(a,b), (c,d)\in A^2$ with $(a,b)\neq (c,d)$ implies that (*) can only hold if both sides equal zero, and hence $g(s_i,s_{-i})=g(s_i',s_{-i})$ for all $s_i,s_i'\in U\!D_i(R_i)$ and all $s_{-i}\in U\!D_{-i}(R_{-i})$. Thus, we have obtained Case (ii).
\end{proof}

\begin{claim}\label{GroupLD2} Suppose for every agent $j$ we have a preference $R_j \in \mathcal{R}_j$ on $A$. Then, for every agent $i$, either
	
	(i) there is for every strategy $s_i\in U\!D(R_i)$ an alternative $a$ such that $g(s_i,s_{-i})=a$ for all $s_{-i}\in U\!D_{-i}(R_{-i})$,
	
	or
	
	(ii) there is for every strategy combination $s_{-i}\in U\!D_{-i}(R_{-i})$ an alternative $a$ such that $g(s_i,s_{-i})=a$ for all $s_i\in U\!D_i(R_i)$,
	
	or both.
\end{claim}

\begin{proof}
	Claim \ref{GroupLD2} follows from Claim \ref{GroupLD1} if we represent for each $j$ with $j\neq i$ the preference $R_j$ by the utility function $u_j^*$ referred to in Claim \ref{Concav} because then: $U\!D_{-i}(u^*_{-i})=U\!D_{-i}(R_{-i})$.
\end{proof}

{\sc Completing the Proof of Theorem \ref{thm:MainR}:} The claim is obviously true if there is an alternative $a$ such that $g(s)=a$ for all $s\in U\!D(R)$. Therefore from now on we restrict attention in this proof to the case that there are two alternatives $a\neq b$ such that  $g(s)=a$ for some $s\in U\!D(R)$ and $g(s')=b$ for some other $s'\in U\!D(R)$.\medskip

We shall say that agent $i\in I$ ``has no influence'' if for every $s_{-i}\in U\!D_{-i}(R_{-i})$ there is an $a\in A$ such that $g(s_i,s_{-i})=a$ for all $s_i\in U\!D_i(R_i)$, and we shall say that agent $i$ is a dictator if  agent $i$ has the property ascribed to agent $i^*$  in Theorem \ref{thm:MainR}.  By Claim \ref{GroupLD2} every agent $i$ either has no influence, or is a dictator.\medskip

Next note that it cannot be that there is more than one dictator. A dictator can enforce any of the alternatives contained in $\{g(s) | s\in U\!D(R)\}$. We have assumed that there are at least two such alternatives, say $a$ and $b$. Having two dictators leads to a contradiction if one of them chooses an action that enforces $a$, and the other one chooses an action that enforces $b$.\medskip

Finally note that it cannot be that all agents have no influence. Recall that we are considering the case in which there are two alternatives $a\neq b$ such that  $g(s)=a$ for some $s\in U\!D(R)$ and $g(s')=b$ for some other $s'\in U\!D(R)$. Consider the sequence of $n$ strategy combinations $s^k$ obtained by switching sequentially first agent 1, then agent 2, etc. from strategy $s_i$ to strategy $s_i'$. Thus, $s^1=(s_1',s_2,\ldots,s_n)$, $s^2=(s_1',s_2',s_3 \ldots,s_n)$, etc. Define $s^0=s$. Because $g(s^0)\neq g(s^n)$, there must be some $k$ such that $g(s^k)\neq g(s^{k-1})$. But this means that by construction agent $k$ has influence. Hence agent $k$ must be a dictator. \qed

\section{Proof of Theorem 2}

\smallskip

Part (1) of Theorem \ref{thm:MM} is obvious. Here, we only provide the proof for part (2). Consider a type 1 strategically simple mechanism, and let:
$$i^* \in \bigcap \limits_{R \, \in \bigtimes_{i\in I} \mathcal{R}_i}I^*(R).$$
We shall show that, for all $i \neq i^*$ and all $R_i \in \mathcal{R}_i$, the set $U\!D_i(R_i)$ contains exactly one element. Suppose that, for some $i$ and $R_i$, the set $U\!D_i(R_i)$ had two distinct elements, say $s_i$ and $s_i'$. Consider any $s_{-i}\in S_{-i}$. We claim that $g(s_i,s_{-i})=g(s_i',s_{-i})$. To see this, first note that $s_{-i} \in U\!D_{-i}(R_{-i})$ for some $R_{-i}\in \bigtimes_{j \neq i} \mathcal{R}_j$, because we assume that every strategy is not weakly dominated for some utility function. Now consider the preference profile $(R_i, R_{-i})$. Since agent $i^*$ is local dictator for this preference profile, for any $s_i^* \in U\!D_{i^*}(R_{i^*})$, there is an $a \in A$ such that: $g(s_{i^*}, s_{-i^*})=a$ for all $s_{-i^*} \in U\!D_{-i^*}(R_{-i^*})$. This implies:  $g(s_{i^*}, s_i,s_{-(i^*,i)})=g(s_{i^*},s_i',s_{-(i^*,i)})$ for all $s_{-(i^*,i)}\in U\!D_{-(i^*,i)}(R_{-(i^*,i)})$. As this holds for all $s_i^* \in U\!D_{i^*}$, the assertion follows. But this contradicts our assumption that mechanisms do not have duplicate strategies.\medskip

Fix any $s_{i^*}\in S_{i^*}$, and consider the mechanism in which we have removed agent $i^*$ from the set of agents, in which all other agents have the same strategy sets as originally, i.e., $S_j$, and in which the outcome corresponding to any $s_{-i^*}$ is given by $g(s_{i^*}, s_{-i^*})$. Let us call this mechanism the ``restricted mechanism'' corresponding to $s_{i^*}$. If all agents $j\neq i$ play the strategies that are uniquely dominant in the overall mechanism, then the restricted mechanism implements an outcome function: $F_{s_{i^*}}: \mathbf{U}_{-i} \times \mathbf{M}_{-i} \rightarrow A$. Because, in the overall mechanism, agents have dominant strategies, the outcome correspondence is constant with respect to beliefs, and it is also constant if utility functions are changed without changing the order of the  elements of $A$. We can therefore write $F$ as: $F_{s_{i^*}}: \bigtimes_{j \neq i^*} \mathcal{R}_j \rightarrow A$. We can treat this outcome function as a direct mechanism. Because agents choose dominant strategies in the overall mechanism, in the direct mechanism it is a dominant strategy for each agent to report their preferences truthfully.\medskip


Because in the overall mechanism agents have uniquely dominant strategies, they must have for every preference ordering a strategy that induces in each of the restricted mechanisms a dominant strategy. Agent $i^*$ thus expects, for each of the strategies that he can choose, the same outcome distribution as he would in the sequential mechanism described in Theorem \ref{thm:MM}, if the second stage mechanisms are the restricted mechanisms described by the outcome function $F_{s_{i^*}}$. Agent $i^*$ will make the same choice as in the sequential mechanism as in the given type 1 strategically simple mechanism. This implies part (ii) of Theorem \ref{thm:MM}. \qed

\section{Further Properties of Strategically Simple Mechanisms} \label{app:structure}

\smallskip

In this section of the appendix we provide some general properties of strategically simple mechanisms that will be useful for the proofs of our results for the bilateral trade and the voting applications. This section's results rely on the domain assumptions in Theorem \ref{thm:MainR}, but they are not restricted to the applications, and therefore also complement the main text's results on strategically simple mechanisms in general.\medskip

Before we turn to our results on strategically simple mechanisms, we record a simple observation that is a direct implication of the notion of weak dominance, and that is unrelated to strategic simplicity.

\begin{lemma} \label{lem:UD}
	Let $i \in I$, $R_i, R_i' \in \mathcal{R}_i$, and $s_i \in U\!D_i (R_i)$. Then there exists $s_i' \in U\!D_i (R_i')$ such that for all $s_{-i} \in S_{-i}$:
	\[
	g(s_i', s_{-i}) R_i' g(s_i, s_{-i}).
	\]
\end{lemma}

\begin{proof}
	If $s_i \in U\!D_i (R_i')$, the lemma is true if we set $s_i' = s_i$. If $s_i \notin U\!D_i (R_i')$, because we are considering finite mechanisms, there is some $s_i'' \in U\!D_i (R_i')$ that weakly dominates $s_i$, and the lemma follows if we set $s_i' = s_i''$.
\end{proof}

For the main results of this section of the appendix, we need some additional notation.  For each $i \in I$ and each $s_{-i} \in S_{-i}$, let:
\begin{align*}
	M_i (s_{-i}) & \equiv \{g(s_i, s_{-i}) | s_i \in S_i\}, \text{ and} \\
	b_i (s_{-i}, R_i) & \equiv \arg\max_{R_i} M_i (s_{-i}).
\end{align*}
In words, $M_i (s_{-i})$ is the set of all possible outcomes when agent $i$'s opponents choose $s_{-i}$, and $b_i (s_{-i}, R_i)$ is agent $i$'s most preferred outcome in the set $M_i (s_{-i})$ if agent $i$'s preference is $R_i$. In this paper, we say that $M_i (s_{-i})$ is the ``menu'' offered to agent $i$ if the other agents choose $s_{-i}$.\medskip

\begin{lemma} \label{lem:structure}
	
	Consider a strategically simple mechanism under the domain assumptions in Theorem \ref{thm:MainR}. Suppose that for some $R_{-i} \in \mathcal{R}_{-i}$ the set $U\!D_{-i} (R_{-i})$ has at least two elements. Let $R_i \in \mathcal{R}_i$.
	
	\begin{itemize}
		
		\item [(1)] Suppose that for some $s_{-i}', s_{-i}'' \in U\!D_{-i} (R_{-i})$ we have $b_i (s_{-i}', R_i) \neq b_i (s_{-i}'', R_i)$. Then for all $s_{-i} \in U\!D_{-i} (R_{-i})$ and all $s_i \in U\!D_i (R_i)$:
		\[
		g(s_i, s_{-i}) = b_i (s_{-i}, R_i).
		\]
		
		\item [(2)] Suppose that for some $s_{-i}', s_{-i}'' \in U\!D_{-i} (R_{-i})$ we have $b_i (s_{-i}', R_i) = b_i (s_{-i}'', R_i) = a$ for some outcome $a$. Then there exists some strategy $s_i \in U\!D_i (R_i)$ such that:
		\[
		g(s_i,s_{-i}') = g(s_i,s_{-i}'') = a.
		\]
		Moreover, for all $s_i \in U\!D_i (R_i)$ we have:
		\[
		g(s_i, s_{-i}')= g(s_i, s_{-i}'').
		\]
		
	\end{itemize}
	
\end{lemma}

In words, this lemma says the following. Consider any agent $i$. Consider any preference profile for which at least one agent other than $i$ has more than one undominated strategy. Each of the profiles of undominated strategies of the other agents offers a menu to agent $i$, and from each of these menus we can determine the outcome that agent $i$ prefers most. In case (1), agent $i$'s preferred alternative is different for two such menus. In this case, \emph{all} undominated strategies of agent $i$ must yield agent $i$'s most preferred alternative from the menu offered by the other agents for all undominated strategy profiles of the other agents. In case (2), agent $i$'s preferred alternative is the same for two such menus. In this case, \emph{some, but not all} undominated strategies of agent $i$ must yield agent $i$'s most preferred alternative from the menu offered by the other agents. In case (2) every other undominated strategy of agent $i$ must, however, yield the same outcome, regardless of which undominated strategy profile the other agents choose.


\begin{proof}
	
	(1) Without loss of generality, let $b_i (s_{-i}', R_i) = a$ and $b_i (s_{-i}'', R_i) = b$. Also without loss of generality, we assume that $R_i$ ranks $a$ above $b$. Therefore, $a \notin M_i (s_{-i}'')$. By Lemma \ref{lem:UD}, agent $i$ has an undominated strategy $s_i'$ such that $g(s_i' , s_{-i}')=a$. Because $a \notin M_i (s_{-i}'')$, we know that $g(s_i', s_{-i}'') \neq a$. This means that agent $i$ is not the local dictator, and hence some agent other than agent $i$ must be the local dictator. That is, for every strategy $s_{-i} \in U\!D_{-i} (R_{-i})$, the outcome is independent of agent $i$'s strategy $s_i \in U\!D_i (R_i)$. Moreover, by Lemma \ref{lem:UD}, the outcome has to be agent $i$'s most preferred outcome in the set $M_i (s_{-i})$.\medskip
	
	(2) Suppose that there is no strategy $s_i \in U\!D_i (R_i)$ such that $g(s_i, s_{-i}') = g(s_i, s_{-i}'') = a$. Since $b_i (s_{-i}', R_i) = a$, by Lemma \ref{lem:UD}, there must be some $s_i' \in U\!D_i (R_i)$ such that $g(s_i', s_{-i}') = a$. Hence, $g(s_i', s_{-i}'') \neq a$. This means that agent $i$ is not the local dictator, and hence some agent other than agent $i$ must be the local dictator. This implies that $g(s_i, s_{-i}'') \neq a$ for all $s_i \in U\!D_i (R_i)$. But this contradicts that $b_i (s_{-i}'', R_i) = a$ and Lemma \ref{lem:UD}. Thus, we can conclude that there exists some strategy $s_i \in U\!D_i (R_i)$ such that $g(s_i, s_{-i}') = g(s_i, s_{-i}'') = a$.\medskip
	
	Now either for all $s_i \in U\!D_i (R_i)$ we also have $g(s_i, s_{-i}') = g(s_i, s_{-i}'') = a$, in which case the second part of the assertion obviously holds, or there exists some undominated strategy that sometimes yields an outcome other than $a$ if her opponents choose $s_{-i}'$ or $s_{-i}''$. In the latter case, agent $i$ must be the local dictator, and the second part of the assertion follows.
\end{proof}

Part (2) of the above lemma has a simple implication. Suppose for some preference profile of the agents other than $i$ at least one agent has two undominated strategies. Then no two profiles of undominated strategies of the agents other than $i$ can offer the same menu to agent $i$.

\begin{corollary} \label{cor:not-same-menu}
	
	Consider a strategically simple mechanism under the domain assumptions of Theorem \ref{thm:MainR}. Suppose that for some $R_{-i} \in \mathcal{R}_{-i}$ the set $U\!D_{-i} (R_{-i})$ has at least two elements. If $s_{-i}', s_{-i}'' \in U\!D_{-i} (R_{-i})$ and $s_{-i}' \neq s_{-i}''$, then:
	\[
	M_i (s_{-i}') \neq M_i (s_{-i}'').
	\]
	
\end{corollary}

\begin{proof}
	The proof is indirect. If $M_i (s_{-i}') = M_i (s_{-i}'')$, then $b_i (s_{-i}', R_i) = b_i (s_{-i}'', R_i)$ for all $R_i \in \mathcal{R}_i$. Therefore, case (2) of Lemma \ref{lem:structure} applies, and for any strategy of agent $i$ that is undominated for some preference (in other words, for all strategies of agent $i$), $s_{-i}'$ and $s_{-i}''$ yield the same outcome. They are thus duplicate strategies. But this contradicts our assumption that the mechanism does not contain any duplicate strategies.
\end{proof}

\section{Proof of Proposition \ref{BT2}} \label{app:bilateral-trade}\smallskip

To simplify the notation, we shall use ``$v_i$'' not just to refer to agent $i$'s value of the object, but also to refer to the corresponding ordinal preference. We use $U\!D_i(v_i)$ to denote the set of strategies of agent $i$ that are not weakly dominated if agent $i$ has ordinal preference $v_i$. We use $I^*(v_S,v_B)$ to denote the set of local dictators at preference profile $(v_S,v_B)$.  Finally, for any $(v_S,v_B)$, we denote by $\mathcal{O}(v_S,v_B)$ the set of outcomes that can arise when both agents play strategies that are not weakly dominated given their valuations. That is, $\mathcal{O}(v_S,v_B) \equiv \{ a \in A \, | \, a = g(s_S,s_B) \text{ for some } s_S \in U\!D_S(v_S) \text{ and } s_B \in U\!D_B(v_B) \}$.\medskip

We now prove four claims that will be useful in the proof of the proposition. These claims describe implications of strategic simplicity in the bilateral trade setting, regardless of whether we are referring to type 1 or type 2 strategic simplicity.

\begin{claim}\label{TDOutcomes}
	If $I^*(v_S,v_B)=\{S,B\}$, then $|\mathcal{O}(v_S,v_B)| =1 $.
\end{claim}

\begin{proof}
	This immediately follows from the definition of local dictatorship: if one agent were able to enforce two different outcomes, then the other agent could not be a local dictator.
\end{proof}

\begin{claim}\label{SDOutcomes}
	If $I^*(v_S,v_B)=\{i^*\}$ for some $i^*\in I$, then $|\mathcal{O}(v_S,v_B)| \geq 2$, and $\mathcal{O}(v_S,v_B)\cap T\neq \emptyset.$
\end{claim}

\begin{proof}
	The first part of the claim follows from the fact that if $\mathcal{O}(v_S,v_B)$ had just one element, then both agents would be local dictators. The second part of the claim is a direct implication of the first part.
\end{proof}

{\sc Claim} \ref{SDOutcomes} implies that the following notations for pairs $(v_S,v_B)$ such that $I^*(v_S,v_B)$ $HagertyRogerson1987=\{i^*\}$ for some $i^*\in I$ are well-defined: $\bar t(v_S,v_B)\equiv\max \mathcal{O}(v_S,v_B)$ $\cap T$ and $\underline t(v_S,v_B)\equiv\min \mathcal{O}(v_S,v_B)\cap T$.\medskip

Next, we show that the assumption that each agent has an opting out strategy implies that only ex post individually rational outcomes can occur when agents do not choose weakly dominated strategies.

\begin{claim}\label{EPIR}
	For any $(v_S,v_B)\in V_S\times V_B$, for every $i\in \{S,B\}$, agent $i$ with preferences $v_i$ weakly prefers every outcome in $\mathcal{O}(v_S,v_B)$ to no trade.
\end{claim}

\begin{proof}
	The claim is straightforward for outcomes when both agents are local dictators. By Lemma \ref{lem:UD}, each agent weakly prefers at least one outcome in $\mathcal{O}(v_S,v_B)$ to no trade. By Claim \ref{TDOutcomes}, if both agents are local dictators, $\mathcal{O}(v_S,v_B)$ has just one element. Hence, both agents must weakly prefer this outcome to no trade.\medskip
	
	In the rest of the proof we focus on the case of a unique local dictator, $I(v_S,v_B)=\{i^*\}$. Consider first the agent who is \emph{not} the local dictator, i.e. agent $i\neq i^*$. Obviously, it is sufficient to consider only outcomes in $\mathcal{O}(v_S,v_B)$ that correspond to trade at some price $t\in T$. Consider any strategy $s_{i^*}\in U\!D_{i^*}(v_{i^*})$ of agent $i^*$ that results in trade at price $t$ against any strategy in $U\!D_{i}(v_{i})$ (see Figure \ref{fig:Prop-illustration-1}). Because $i$ has an opting out strategy, by Lemma \ref{lem:UD} one of the strategies in $U\!D_i(v_i)$ must yield at least as good an outcome as no trade for agent $i$ with preferences $v_i$. This implies that trade at price $t$ must be at least as good as no trade for agent $i$ with preference $v_i$. \smallskip
	
	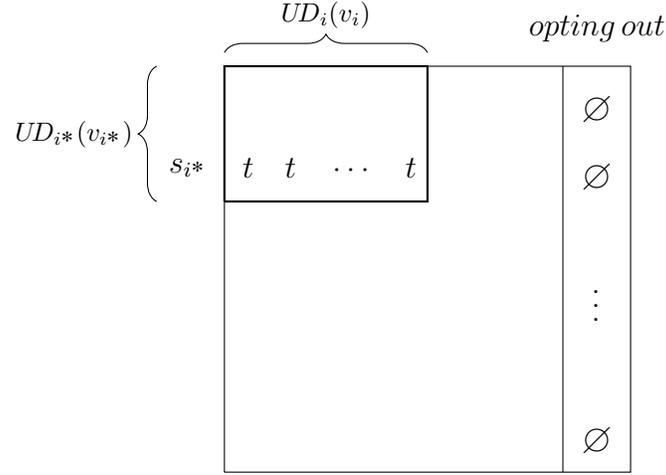
\begin{figure}[h]
		\label{fig-agent-not-the-local-dictator}
		\centering
		\begin{tikzpicture}[scale=0.9]
		\draw (0,0) -- (0,6) -- (6,6) -- (6,0) -- (0,0);
		\draw [thick] (0,4) -- (0,6) -- (3,6) -- (3,4) -- (0,4);
		\draw (5,0) -- (5,6);
		\node [left] at (-0.1,4.5) {$s_{i^*}$};
		\node [left] at (3,4.5) {$t \quad t \quad \cdots \quad t$};
		\node [above] at (5.5,6.2) {$opting \, out$};
		\node [above] at (5.5,5) {$\emptyset$};
		\node [above] at (5.5,4) {$\emptyset$};
		\node[rotate=90,yshift=0pt] at (5.5,2.5) {$\cdots$};
		\node [above] at (5.5,0.1) {$\emptyset$};
		\draw [decorate,decoration={brace,amplitude=7pt},xshift=0pt,yshift=0pt]
		(-1,4) -- (-1,6) node [black,midway,xshift=-1.1cm]
		{\footnotesize $U\!D_{i^*}(v_{i^*})$};
		\draw [decorate,decoration={brace,amplitude=7pt},xshift=0pt,yshift=0pt]
		(0,6.2) -- (3,6.2) node [black,midway,yshift=0.5cm]
		{\footnotesize $U\!D_i(v_i)$};
		\end{tikzpicture}
		\caption{\small  Agent $i^*$ is the unique local dictator at $(v_S,v_B)$. In this case, trade at price $t$ must not be worse than no trade for agent $i$ with value $v_i$.} \label{fig:Prop-illustration-1}
	\end{figure} \smallskip
	
	We are left with the task to show that, when there is a unique local dictator $i^*$, all outcomes are ex post individually rational for the local dictator herself. Without loss of generality we consider the case $i^*=S$. Our proof strategy will be the following. We consider any $(v_S,v_B)$ such that ex post individual rationality for the seller is violated at $(v_S,v_B)$. We  show that then there must be some $v_S'>v_S$ and some $v_B'$ such that $I^*(v_S',v_B')=\{S\}$ and ex post individual rationality for the seller is also violated at $(v_S',v_B')$. This implies the claim, because the assumption that there is any value profile at which the seller's ex post individual rationality were violated would imply that there would have to be a largest $v_S\in I^*_S$ for which individual rationality is violated for some $v_B$, and this would be in contradiction with the assertion that we just made.\medskip
	
	Thus, consider any $(v_S,v_B)$ such that $I^*(v_S,v_B)=\{S\}$ and the seller's individual rationality is violated at $(v_S,v_B)$ (see Figure \ref{fig:Prop-illustration-2}). This means that there is a strategy $s_S\in U\!D_S(v_S)$ for which  $g(s_S,s_B)\in T$, and $g(s_S,s_B)<v_S$ for all $s_B \in U\!D_B(v_B)$. To start, note that there must exist some $v_B' \in V_B$ and $s_B' \in U\!D_B(v_B')$ such that $v_S$ ranks $g(s_S, s_B')$ above no trade. Otherwise, for the seller with preference $v_S$, the strategy $s_S$ would be weakly dominated by the strategy of opting out.  Since $v_S$ ranks $g(s_S, s_B')$ above no trade, we have $g(s_S, s_B') \in T$ and $g(s_S, s_B') > v_S$.\smallskip
	
	\begin{figure}[h]
		\label{fig-agent-the-local-dictator}
		\centering
		\begin{tikzpicture}[scale=0.9]
		\draw (0,0) -- (0,11) -- (11,11) -- (11,0) -- (0,0);
		\draw [thick] (0,7) -- (0,11) -- (4,11) -- (4,7) -- (0,7);
		\draw [thick] (5,7) -- (5,11) -- (9,11) -- (9,7) -- (5,7);
		\draw [thick] (5,2) -- (5,6) -- (9,6) -- (9,2) -- (5,2);
		\draw (10,0) -- (10,11);
		\draw (0,1) -- (11,1);
		\node [left] at (-0.1,7.5) {$s_S$};
		\node [above] at (0.5,11.2) {$s_B$};
		\node [above] at (6,11.2) {$s_B'$};
		\node [above] at (7.3,11.2) {$s_B''$};
		\node [left] at (-0.1,2.5) {$s_S'$};
		\node [above] at (6,10.2) {\footnotesize $g(s_S,s_B')$};
		\node [above] at (6,9.2) {\footnotesize $g(s_S,s_B')$};
		\node [above] at (6,7.2) {\footnotesize $g(s_S,s_B')$};
		\node [rotate=270,yshift=0pt] at (6,8.5) {$\cdots$};
		
		\node [above] at (7.3,10.2) {$\emptyset$};
		\node [above] at (7.3,9.2) {$\emptyset$};
		\node [above] at (7.3,7.2) {$\emptyset$};
		\node [rotate=270,yshift=0pt] at (7.3,8.5) {$\cdots$};
		
		\node [above] at (10.5,11.2) {$opting \, out$};
		\node [above] at (10.5,10.2) {$\emptyset$};
		\node [above] at (10.5,9.2) {$\emptyset$};
		\node[rotate=90,yshift=0pt] at (10.5,5.5) {$\cdots$};
		\node [above] at (10.5,0.2) {$\emptyset$};
		
		\node [left] at (-0.2,0.5) {$opting \, out$};
		\node [above] at (0.5,0.2) {$\emptyset$};
		\node [above] at (1.5,0.2) {$\emptyset$};
		\node [above] at (5.5,0.2) {$\cdots$};
		
		\node [left] at (4.1,7.5) {\footnotesize $g(s_S,s_B) \, \cdots \, g(s_S,s_B)$};
		\node [left] at (9.1,2.5) {\footnotesize $g(s_S',s_B') \, \cdots \, g(s_S',s_B')$};
		
		\draw [decorate,decoration={brace,amplitude=7pt},xshift=0pt,yshift=0pt]
		(-0.8,7) -- (-0.8,11) node [black,midway,xshift=-1cm]
		{\footnotesize $U\!D_S(v_S)$};
		
		\draw [decorate,decoration={brace,amplitude=7pt},xshift=0pt,yshift=0pt]
		(-0.8,2) -- (-0.8,6) node [black,midway,xshift=-1cm]
		{\footnotesize $U\!D_S(v_S')$};
		
		\draw [decorate,decoration={brace,amplitude=7pt},xshift=0pt,yshift=0pt]
		(0,11.9) -- (4,11.9) node [black,midway,yshift=0.5cm]
		{\footnotesize $U\!D_B(v_B)$};
		
		\draw [decorate,decoration={brace,amplitude=7pt},xshift=0pt,yshift=0pt]
		(5,11.9) -- (9,11.9) node [black,midway,yshift=0.5cm]
		{\footnotesize $U\!D_B(v_B')$};
		\end{tikzpicture}
		\caption{\small The seller is the unique local dictator at $(v_S,v_B)$. Suppose that $g(s_S,s_B) < v_S$. We then find a  $v_S' > v_S$ such that $g(s_S',s_B') < v_S'$.} \label{fig:Prop-illustration-2}
	\end{figure} \smallskip
	
	Our next objective is to prove the following statements about the behavior of the mechanism at $(v_S,v_B)$ and $(v_S,v_B')$. Here, $s_B$ is any arbitrary strategy in $U\!D_B(v_B)$.
	\begin{itemize}
		\item[(i)] $B$ is the unique local dictator at $(v_S,v_B')$;
		\item[(ii)] $g(s_S,s_B')>g(s_S,s_B)$;
		\item[(iii)] $g(s_S,s_B)>v_B'$;
		\item[(iv)] $g(s_S,s_B')>v_B'$.
	\end{itemize}\medskip
	
	Proving $(ii)$ is simple: We have $g(s_S,s_B)<v_S$, and, by construction, $g(s_S,s_B')>v_S$. Thus, $(ii)$ follows. Now note that $g(s_S, s_B') >g(s_S,s_B)$ implies that $v_B'$ ranks $g(s_B,s_B')$ below $g(s_S,s_B)$. By Lemma \ref{lem:UD}, there must be some strategy $s_B'' \in U\!D_B(v_B')$ such that $v_B'$ ranks $g(s_S, s_B'')$ above $g(s_S,s_B)$ or $g(s_S, s_B'') = g(s_S,s_B)$. Note that we can conclude $g(s_S, s_B)\neq g(s_S, s_B'')$, and hence that $(i)$ is true.\medskip
	
	As an intermediate step we show next that $g(s_S, s_B'')=\phi$. If $g(s_S, s_B'')$ were an element of $T$, since $v_B'$ ranks $g(s_S, s_B'')$ above $g(s_S,s_B)$ or $g(s_S, s_B'') = g(s_S,s_B)$, it would have to be that $g(s_S, s_B'') \leq g(s_S,s_B)$. Since $v_S$ ranks $g(v_S,v_B)$ below no trade, $v_S$ also ranks $g(s_S, s_B'')$ below no trade. But this contradicts the ex post individual rationality for the agent who is not dictator, which we showed in an earlier step of this proof.  We conclude: $g(s_S, s_B'')=\phi$.\medskip
	
	By construction, $v_B'$ ranks $g(s_S, s_B'')$ above $g(s_S,s_B)$, hence $(iii)$ follows from the fact that
	$g(s_S, s_B'')$ is no trade. Finally, $(ii)$ and $(iii)$ imply $(iv)$.\medskip
	
	Now note that we have obtained a pair of valuations at which the buyer is the local dictator, and, by $(iv)$, the buyer's ex post individual rationality is violated. We can therefore repeat the argument just presented, reversing the roles of the buyer and the seller. This yields the conclusion that there is some $v_S'\in V_S$, and some $s_S' \in U\!D(v_S')$ such that $g(s_S', s_B')\in T$ and $g(s_S',s_B')<v_B'$, and:
	\begin{itemize}
		\item[(v)] $S$ is the local dictator at $(v_S',v_B')$;
		\item[(vi)] $g(s_S',s_B')<g(s_S,s_B')$;
		\item[(vii)] $g(s_S,s_B')<v_S'$;
		\item[(viii)] $g(s_S',s_B')<v_S'$.
	\end{itemize} \smallskip \smallskip
	
	The proof can now be concluded. By construction: $v_S<g(s_S,s_B')$. Result $(vii)$ says: $g(s_S,s_B')<v_S'$. Hence $v_S<v_S'$. Moreover $(viii)$ shows that ex post individual rationality for the seller is violated at $(v_S',v_B')$.
\end{proof}

Our next result shows that, if at some valuation profile  some agent $i$ is the unique local dictator, this agent remains (not necessarily unique) local dictator even if we change $i$'s valuation, keeping the other valuation fixed.

\begin{claim}\label{SDSwitchSD}
	Suppose $I^*(v_i,v_{-i})=\{i\}$. Then $i\in I^*(v_i',v_{-i})$ for all $v_i'\in V_i$.
\end{claim}

\begin{proof}
	Without loss of generality we focus on the case $i=S$. The proof is indirect. Let $I^*(v_S,v_B)=\{S\}$, and suppose $I^*(v_S, v_B)=\{B\}$ for some $v_S' \in V_S$. Let $s_S \in S_S$ be the strategy in $U\!D_S(v_S)$ that enforces the outcome $\bar t(v_S,v_B)$ against any strategy in $U\!D_B(v_B)$. Let $s_B \in S_B$ be the strategy in $U\!D_B(v_B)$ that enforces the outcome $\bar t(v_S', v_B)$ against any strategy in $U\!D_S(v_S')$.\medskip
	
	Suppose also, first, that: $\bar t(v_S',v_B)> \bar t(v_S,v_B)$. By Lemma \ref{EPIR}, $v_S$ ranks $\bar t(v_S,v_B)$ above no trade. Therefore, $v_S$ must also rank $\bar t(v_S',v_B)$ above no trade. By Lemma \ref{lem:UD}, the seller with value $v_S$ must have a strategy in $U\!D_S(v_S)$ that guarantees an outcome at least as good as  $\bar t(v_S',v_B)$ against any strategy in $U\!D_B(v_B)$. This contradicts the definition of $\bar t(v_S,v_B)$ as the highest price that the seller can guarantee with a strategy in $U\!D_S(v_S)$.\medskip
	
	Now suppose: $\bar t(v_S',v_B)< \bar t(v_S,v_B)$. By Lemma \ref{lem:UD}, the seller with value $v_S'$ must have at least one strategy in $U\!D_S(v_S')$ that yields against $s_B$ an outcome at least as good as $\bar t(v_S,v_B)$. This contradicts that $s_B$ yields $\bar t(v_S',v_B)$ for all $s_S\in U\!D_S(v_S)$.\medskip
	
	Finally suppose: $\bar t(v_S',v_B)= \bar t(v_S,v_B)$. Let $s_B' \in U\!D_B(v_B)$ denote a strategy such that $g(s_S, s_B') \neq \bar t(v_S',v_B)$ for all $s_S \in U\!D_S(v_S')$. By Claim \ref{SDOutcomes}, such an $s_B'$ exists. Since the buyer is the unique local dictator at preference profile $(v_S', v_B)$, by Lemma \ref{lem:UD}, any strategy in $U\!D_S(v_S')$ yields against $s_B'$ an outcome at least as good as $\bar t(v_S,v_B)$. This outcome cannot be trade at price $\bar t(v_S,v_B)$, because $s_B'$ leads to an outcome other than $\bar t(v_S,v_B)$, and it cannot be trade at a price higher than $\bar t(v_S,v_B)$ because we are considering the case $\bar t(v_S',v_B)= \bar t(v_S,v_B)$. Therefore, any strategy in $U\!D_S(v_S')$ yields against $s_B'$ no trade. But then we have concluded that the seller prefers no trade to trade at $\bar t(v_S,v_B)$, which contradicts Claim \ref{EPIR}, i.e. the seller's ex post individual rationality at $(v_S,v_B)$ and at $(v_S',v_B)$. \end{proof}

We now turn to an indirect proof of Proposition \ref{BT2}, that is, we postulate that a bilateral trade mechanism is  type 2 strategically simple, and then derive a contradiction. The next four claims describe implications of the premises of the indirect proof.

\begin{claim}\label{TwoSD}
	There are $v_S,\hat v_S\in V_S$ with $v_S\neq \hat v_S$ and $v_B,\hat v_B\in V_B$ with $v_B\neq \hat v_B$ such that: $I^*(v_S,v_B)=\{S\}, I^*(\hat v_S, \hat v_B)=\{B\}$, and $I^*(v_S, \hat v_B)=I^*(\hat v_S, v_B)=\{S,B\}$.
\end{claim}

\begin{proof}
	By definition of type 2 strategic simplicity, we must have two pairs of values in $V_S\times V_B$, one at which $S$ is the unique local dictator, and another one at which $B$ is the unique local dictator. By Claim \ref{SDSwitchSD} these two pairs must have no component in common. Claim \ref{SDSwitchSD} also implies that if we combine the seller's value in one pair with a buyer's value in the other pair, then both agents must be local dictators.
\end{proof}

For the remainder of the proof we use the notation $(v_S, v_B)$ and $(\hat v_S, \hat v_B)$ to refer to the two pairs the existence of which is asserted in Claim \ref{TwoSD}.

\begin{claim}\label{NoTrade}
	$\mathcal{O}(v_S,\hat v_B)=\{\phi\}$.
\end{claim}

\begin{proof}
	By Claim \ref{TDOutcomes}, $\mathcal{O}(v_S,\hat v_B)$ has only one element. Suppose $\mathcal{O}(v_S,\hat v_B)=\{t\}$ for some $t\in T$. Using Lemma \ref{lem:UD} for the buyer, we can infer $t\leq \underline t(v_S, v_B)$. Because at $(v_B,v_S)$ the seller is the only local dictator, Claim \ref{SDOutcomes} implies that the set $\mathcal{O}(v_S,v_B)$ must include an outcome $a$ other than $\underline t(v_S,v_B)$. If this is trade at a price higher than $\underline t(v_S,v_B)$, then clearly the buyer strictly prefers $\underline t(v_S,v_B)$ to $a$. But if $a$ is no trade, then Claim \ref{EPIR} implies that the buyer strictly prefers $\underline t(v_S,v_B)$ to $a$. Thus, $\mathcal{O}(v_S,v_B)$ includes an outcome $a$  that the buyer ranks  strictly below $\underline t(v_S,v_B)$, and hence also strictly below $t$.  The seller has a strategy that locally enforces this outcome at $(v_S,v_B)$. By Lemma \ref{lem:UD} this contradicts the fact that the buyer has a strategy that enforces at $(v_S,\hat v_B)$ the price $t$.
\end{proof}

\begin{claim}\label{ValueComparisons}
	$v_S>\hat v_S$ and $v_B>\hat v_B$.
\end{claim}

\begin{proof}
	
The arguments are symmetric for seller and buyer. Consider the seller. Because no trade occurs at $(v_S,\hat v_B)$, by Claim \ref{NoTrade}, and trade at some price $t$ is a possible outcome at $(\hat v_S, \hat v_B)$, Lemma \ref{lem:UD} implies that with value $v_S$ the seller must find no trade preferable to a trade at  price $t$. Claim \ref{EPIR} says that the seller with value $\hat v_S$ prefers trade at price $t$ to no trade. These findings together imply $v_S > \hat v_S$.
\end{proof}

\begin{claim}\label{Trade}
	$\mathcal{O}(\hat v_S, v_B)=\{t^*\}$ for some $t^*\in T$.
\end{claim}

\begin{proof}
	By Claim \ref{TDOutcomes}, $\mathcal{O}(\hat v_S,v_B)$ has only one element. Suppose $\mathcal{O}(\hat v_S,v_B) = \{\phi\}$. By Claims \ref{SDOutcomes} and \ref{EPIR}, trade at some price is contained in $\mathcal{O}(v_S,v_B)$ that  the seller with value $v_S$ strictly prefers to no trade. When the seller has value $\hat v_S$, the seller still strictly prefers trade at that price to no trade, because,  by Claim \ref{ValueComparisons}, $\hat v_S$ is lower than $v_S$. Hence we would have a contradiction to Lemma \ref{lem:UD} if the  outcome in $\mathcal{O}(\hat v_S, v_B)$ were no trade.
\end{proof}

We can now complete the proof of Proposition \ref{BT2}. Using Lemma \ref{lem:UD} we have: $t^* = \underline t(\hat v_S,\hat v_B)$. By Claim \ref{EPIR}, $t^*\leq \hat v_B$. Using Lemma \ref{lem:UD} we also have: $t^*=\bar t(v_S, v_B)$. But then Lemma \ref{lem:UD} and $t^*\leq \hat v_B$ implies that among the outcomes in $\mathcal{O}(\hat v_B, v_S)$ there must be a trade at a price below $\hat v_B$. This contradicts Claim \ref{NoTrade}. \qed

\section{Proof of Proposition \ref{prop:voting-type-2}} \label{app:voting}\smallskip

Let $A = \{a, b, c\}$. Suppose that a mechanism is a strategically simple mechanism of type 2. We shall analyze properties of such a mechanism that ultimately imply that, up to relabeling of the agents and the alternatives, only the two mechanisms listed in Proposition \ref{prop:voting-type-2} are candidates for type 2 strategically simple mechanisms. The analysis of these two mechanisms in the main text shows that these mechanisms are indeed type 2 strategically simple.\medskip

Throughout this proof, we shall denote the ordinal preference $R_i$ that satisfies $a R_i b$ and $b R_i c$ by ``$abc$," and we shall use analogous notation for any other ordinal preference over the three alternatives.

\begin{claim}\label{TwoUD}
There is at least one preference profile $(\hat R_1, \hat R_2)$ such that both $U\!D_1 (\hat R_1)$ and $U\!D_2(\hat R_2)$ have at least two elements.
\end{claim}

\begin{proof}
 At a preference profile at which agent 1 is the unique local dictator, agent 1 must have at least two undominated strategies. At a preference profile at which agent 2 is the unique local dictator, agent 2 must have at least two undominated strategies.
\end{proof}

\begin{claim}\label{No3}
If for some preference profile $(\hat R_1, \hat R_2)$ both $U\!D_1 (\hat R_1)$ and $U\!D_2(\hat R_2)$ have at least two elements, then the set $g(U\!D_1 (\hat R_1), U\!D_2 (\hat R_2))$ has no more than two elements.
\end{claim}

\begin{proof}
By Theorem \ref{thm:MainR}, there must be a local dictator at $(\hat R_1, \hat R_2)$. Without loss of generality, assume that agent 2 is a local dictator. If the set $g(U\!D_1 (\hat R_1), U\!D_2 (\hat R_2))$ contains three elements, then agent 2, as a local dictator, could enforce each of them. Therefore, each of agent 1's undominated strategies would have to offer the same menu that contains all three elements. But this contradicts Corollary \ref{cor:not-same-menu}. Therefore, $g(U\!D_1(\hat R_1),U\!D(\hat R_2))$ has only one or two elements.
\end{proof}

We now distinguish the two cases. Case 1 is the case in which there is at least one preference profile such that both agents have multiple undominated strategies, and such that exactly two outcomes may result if both agents with these preferences choose from their sets of undominated strategies.  For this case, we show that the $4 \times 4$ mechanism in Proposition \ref{prop:voting-type-2} is the unique strategically simple mechanism, up to relabeling of the agents and the alternatives. Case 2 is the case in which for all preference profiles such that both agents have multiple undominated strategies, exactly one outcome may result if both agents with these preferences choose from their sets of undominated strategies. For this case, we show that the $5\times 5$ mechanism in Proposition \ref{prop:voting-type-2} is the unique strategically simple mechanism, up to relabeling of the agents and the alternatives. \bigskip

\noindent{\sc Case 1:} There is at least one preference profile, say $(\hat R_1, \hat R_2)$, such that both agents have multiple undominated strategies, and such that exactly two outcomes may result, say $g(U\!D_1(\hat R_1),U\!D(\hat R_2))=\{a,b\}$, if both agents with these preferences choose from their sets of undominated strategies. \bigskip

Figure \ref{fig:case1} illustrates the proof for the first case. We shall refer to Figure \ref{fig:case1} while presenting the proof.\medskip

\begin{figure}[h]

\centering
	\begin{tikzpicture}[scale=0.55]
	
	\subfloat[Step (1)]{
		\node [left] at (-5, 5) {$\hat{s}_1$};
		\node [left] at (-5, 4) {$\hat{\hat{s}}_1$};
		\node [left] at (-3.9, 6) {$\hat{s}_2$};
		\node [left] at (-2.8, 6) {$\hat{\hat{s}}_2$};
		\node [left] at (-4, 5) {$a$};
		\node [left] at (-4, 4) {$a$};
		\node [left] at (-3, 5) {$b$};
		\node [left] at (-3, 4) {$b$};
		\draw [decorate,decoration={brace,amplitude=7pt},xshift=0pt,yshift=0pt]
		(-6, 3.5) -- (-6, 5.5) node [black,midway,xshift=-1.1cm]
		{\footnotesize $U\!D_1 (\hat R_1)$};
	}
	
	\node [left] at (-1.5, 4.5) {$\Rightarrow$};
	
	\subfloat[Step (2)]{
		\node [left] at (3, 5) {$\hat{s}_1$};
		\node [left] at (3, 4) {$\hat{\hat{s}}_1$};
		\node [left] at (4.1, 6) {$\hat{s}_2$};
		\node [left] at (5.2, 6) {$\hat{\hat{s}}_2$};
		\node [left] at (6.7, 6) {$s_2^{cab}$};
		\node [left] at (8, 6) {$s_2^{cba}$};
		\node [left] at (4, 5) {$a$};
		\node [left] at (4, 4) {$a$};
		\node [left] at (5, 5) {$b$};
		\node [left] at (5, 4) {$b$};
		\node [left] at (6.2, 5) {$c$};
		\node [left] at (6.2, 4) {$a$};
		\node [left] at (7.3, 5) {$c$};
		\node [left] at (7.3, 4) {$b$};
		\draw [decorate,decoration={brace,amplitude=7pt},xshift=0pt,yshift=0pt]
		(2, 3.5) -- (2, 5.5) node [black,midway,xshift=-1.1cm]
		{\footnotesize $U\!D_1 (bca)$};
	}
	
	\node [left] at (9, 4.5) {$\Rightarrow$};
	
	\subfloat[Step (3)]{
		\node [left] at (13.5, 5) {$\hat{s}_1$};
		\node [left] at (13.5, 4) {$\hat{\hat{s}}_1$};
		\node [left] at (14.6, 6) {$\hat{s}_2$};
		\node [left] at (15.7, 6) {$\hat{\hat{s}}_2$};
		\node [left] at (17.2, 6) {$s_2^{cab}$};
		\node [left] at (18.5, 6) {$s_2^{cba}$};
		\node [left] at (14.5, 5) {$a$};
		\node [left] at (14.5, 4) {$a$};
		\node [left] at (15.5, 5) {$b$};
		\node [left] at (15.5, 4) {$b$};
		\node [left] at (16.7, 5) {$c$};
		\node [left] at (16.7, 4) {$a$};
		\node [left] at (17.8, 5) {$c$};
		\node [left] at (17.8, 4) {$b$};
		\node [left] at (14.5, 3) {$a$};
		\node[rotate=90,yshift=0pt] at (14.1, 1.8) {$\cdots$};
		\draw [decorate,decoration={brace,amplitude=7pt},xshift=0pt,yshift=0pt]
		(12.5, 3.5) -- (12.5, 5.5) node [black,midway,xshift=-1.1cm]
		{\footnotesize $U\!D_1(bca)$};
	}
	
	\node [left] at (-6.5, -3.5) {$\Rightarrow$};
	
	\subfloat[Step (5)]{
		\node [left] at (-2, -3) {$\hat{s}_1$};
		\node [left] at (-2, -4) {$\hat{\hat{s}}_1$};
		\node [left] at (-0.9, -2) {$\hat{s}_2$};
		\node [left] at (0.2, -2) {$\hat{\hat{s}}_2$};
		\node [left] at (1.7, -2) {$s_2^{cab}$};
		\node [left] at (3, -2) {$s_2^{cba}$};
		\node [left] at (-1, -3) {$a$};
		\node [left] at (-1, -4) {$a$};
		\node [left] at (0, -3) {$b$};
		\node [left] at (0, -4) {$b$};
		\node [left] at (1.2, -3) {$c$};
		\node [left] at (1.2, -4) {$a$};
		\node [left] at (2.3, -3) {$c$};
		\node [left] at (2.3, -4) {$b$};
		\node [left] at (-1, -5) {$a$};
		\node[rotate=90,yshift=0pt] at (-1.4, -6.2) {$\cdots$};
		\draw [decorate,decoration={brace,amplitude=7pt},xshift=0pt,yshift=0pt]
		(-3, -4.5) -- (-3, -2.5) node [black,midway,xshift=-1.1cm]
		{\footnotesize $U\!D_1(bca)$};
		\draw [decorate,decoration={brace,amplitude=7pt},xshift=0pt,yshift=0pt]
		(-2,-1.5) -- (0,-1.5) node [black,midway,yshift=0.5cm]
		{\footnotesize $U\!D_2 (bac)$};
	}
	
	\node [left] at (4.5, -3.5) {$\Rightarrow$};
	
	\subfloat[Step (5)]{
		\node [left] at (9, -3) {$\hat{s}_1$};
		\node [left] at (9, -4) {$\hat{\hat{s}}_1$};
		\node [left] at (10.1, -2) {$\hat{s}_2$};
		\node [left] at (11.2, -2) {$\hat{\hat{s}}_2$};
		\node [left] at (12.7, -2) {$s_2^{cab}$};
		\node [left] at (14, -2) {$s_2^{cba}$};
		\node [left] at (15.3, -2) {$s_2^{bca}$};
		\node [left] at (10, -3) {$a$};
		\node [left] at (10, -4) {$a$};
		\node [left] at (11, -3) {$b$};
		\node [left] at (11, -4) {$b$};
		\node [left] at (12.2, -3) {$c$};
		\node [left] at (12.2, -4) {$a$};
		\node [left] at (13.3, -3) {$c$};
		\node [left] at (13.3, -4) {$b$};
		\node [left] at (14.5, -3) {$b$};
		\node [left] at (14.5, -4) {$b$};
		\node [left] at (10, -5) {$a$};
		\node[rotate=90,yshift=0pt] at (9.6, -6.2) {$\cdots$};
		\draw [decorate,decoration={brace,amplitude=7pt},xshift=0pt,yshift=0pt]
		(8, -4.5) -- (8, -2.5) node [black,midway,xshift=-1.1cm]
		{\footnotesize $U\!D_1(bca)$};
		\draw [decorate,decoration={brace,amplitude=7pt},xshift=0pt,yshift=0pt]
		(9,-1.5) -- (11,-1.5) node [black,midway,yshift=0.5cm]
		{\footnotesize $U\!D_2 (bac)$};
	}
	
	\node [left] at (-6.5, -11.5) {$\Rightarrow$};
	
	\subfloat[Step (5)]{
		\node [left] at (-2, -11) {$\hat{s}_1$};
		\node [left] at (-2, -12) {$\hat{\hat{s}}_1$};
		\node [left] at (-0.9, -10) {$\hat{s}_2$};
		\node [left] at (0.2, -10) {$\hat{\hat{s}}_2$};
		\node [left] at (1.7, -10) {$s_2^{cab}$};
		\node [left] at (3, -10) {$s_2^{cba}$};
		\node [left] at (4.3, -10) {$s_2^{bca}$};
		\node [left] at (-1, -11) {$a$};
		\node [left] at (-1, -12) {$a$};
		\node [left] at (0, -11) {$b$};
		\node [left] at (0, -12) {$b$};
		\node [left] at (1.2, -11) {$c$};
		\node [left] at (1.2, -12) {$a$};
		\node [left] at (2.3, -11) {$c$};
		\node [left] at (2.3, -12) {$b$};
		\node [left] at (3.5, -11) {$b$};
		\node [left] at (3.5, -12) {$b$};
		\node [left] at (-1, -13) {$a$};
		\node [left] at (0, -13) {$c$};
		\node [left] at (1.2, -13) {$c$};
		\node [left] at (2.3, -13) {$c$};
		\node [left] at (3.5, -13) {$c$};
		\node [left] at (-1, -14) {$a$};
		\node [left] at (0, -14) {$a$};
		\node [left] at (1.2, -14) {$a$};
		\node [left] at (2.3, -14) {$a$};
		\node [left] at (3.5, -14) {$a$};
		\draw [decorate,decoration={brace,amplitude=7pt},xshift=0pt,yshift=0pt]
		(-3, -12.5) -- (-3, -10.5) node [black,midway,xshift=-1.1cm]
		{\footnotesize $U\!D_1(bca)$};
		\draw [decorate,decoration={brace,amplitude=7pt},xshift=0pt,yshift=0pt]
		(-2,-9.5) -- (0,-9.5) node [black,midway,yshift=0.5cm]
		{\footnotesize $U\!D_2 (bac)$};
	}
	
	\node [left] at (5.5, -11.5) {$\Rightarrow$};
	
	\subfloat[Step (5)]{
		\node [left] at (10, -11) {$\hat{s}_1$};
		\node [left] at (10, -12) {$\hat{\hat{s}}_1$};
		\node [left] at (11.1, -10) {$\hat{s}_2$};
		\node [left] at (12.2, -10) {$\hat{\hat{s}}_2$};
		\node [left] at (13.7, -10) {$s_2^{cab}$};
		\node [left] at (15, -10) {$s_2^{cba}$};
		\node [left] at (11, -11) {$a$};
		\node [left] at (11, -12) {$a$};
		\node [left] at (12, -11) {$b$};
		\node [left] at (12, -12) {$b$};
		\node [left] at (13.2, -11) {$c$};
		\node [left] at (13.2, -12) {$a$};
		\node [left] at (14.3, -11) {$c$};
		\node [left] at (14.3, -12) {$b$};
		\node [left] at (11, -13) {$a$};
		\node [left] at (12, -13) {$c$};
		\node [left] at (13.2, -13) {$c$};
		\node [left] at (14.3, -13) {$c$};
		\node [left] at (11, -14) {$a$};
		\node [left] at (12, -14) {$a$};
		\node [left] at (13.2, -14) {$a$};
		\node [left] at (14.3, -14) {$a$};
		\draw [decorate,decoration={brace,amplitude=7pt},xshift=0pt,yshift=0pt]
		(9, -12.5) -- (9, -10.5) node [black,midway,xshift=-1.1cm]
		{\footnotesize $U\!D_1(bca)$};
		\draw [decorate,decoration={brace,amplitude=7pt},xshift=0pt,yshift=0pt]
		(10,-9.5) -- (12,-9.5) node [black,midway,yshift=0.5cm]
		{\footnotesize $U\!D_2 (bac)$};
	}
	
	\end{tikzpicture}
	\caption{\small There is a unique type 2 strategically simple mechanism (up to relabeling) in {\sc Case 1}.} \label{fig:case1}
\end{figure}

We begin our analysis of this case with the observation that with preference $\hat R_1$ agent 1 has only two undominated strategies.

\begin{claim}\label{twomenus}
$U\!D_1(\hat R_1)$ has exactly two elements, one, which we shall denote by $\hat s_1$, satisfies $M_2(\hat s_1)=\{a,b,c\}$, and the other one, which we shall denote by $\hat{\hat s}_1$, satisfies $M_2(\hat{\hat s}_1)=\{a,b\}$.
\end{claim}

\begin{proof}
Because agent 2 is a local dictator at $(\hat R_1,\hat R_2)$, every undominated strategy of agent 1 has to offer a menu that includes both $a$ and $b$. There are only two such menus: $\{a, b, c\}$ and $\{a, b\}$. Because agent 1 has multiple undominated strategies and each such strategy by Corollary \ref{cor:not-same-menu} has to offer a different menu, she has exactly two undominated strategies with one strategy offering menu $\{a, b, c\}$ and the other strategy offering menu $\{a, b\}$.
\end{proof}

Next, we investigate agent 2's strategy set, and for each of her strategies the outcome that results if agent 1 chooses $\hat s_1$ or $\hat{\hat{s}}_1$. Define:
\begin{align*}
	S_2^a & = \{s_2 \in S_2 : g(s_1, s_2) = a \text{ for all } s_1 \in U\!D_1 (\hat R_1)\}, \text{ and } \\
	S_2^b & = \{s_2 \in S_2 : g(s_1, s_2) = b \text{ for all } s_1 \in U\!D_1 (\hat R_1)\}.
\end{align*}

Because, by assumption, in Case 1: $g(U\!D_1(\hat R_1),U\!D_2(\hat R_2))=\{a,b\},$ and because, also by assumption, agent 2 is the local dictator at $(\hat R_1,\hat R_2)$, there must be at least one strategy in $U\!D_2(\hat R_2)$ that is in $S_2^a$, and also at least one strategy in $U\!D_2(\hat R_2)$ that is in $S_2^b$. Let us denote the former strategy by $\hat s_2$ and the latter by $\hat{\hat{s}}_2$. We also know that \emph{all} strategies in $U\!D_2(\hat R_2)$ are contained in $S_2^a\cup S_2^b$. That is because agent 2 is the local dictator at $(\hat R_1,\hat R_2)$. The top left panel in Figure \ref{fig:case1} represents, symbolically, what we have inferred so far about the mechanism that we are considering.\medskip

The focus of {\sc Claims} \ref{agent2strategies} and \ref{agent2strategiesAgain} will be strategies of agent 2 that are \emph{not} in $S_2^a\cup S_2^b$. We shall conclude that there are exactly two such strategies, and we shall show which outcomes they yield against $\hat s_1$ and $\hat{\hat{s}}_1$.

\begin{claim}\label{agent2strategies}
If $s_2\in S_2 \setminus (S_2^a \cup S_2^b)$, then either:
\[
g(\hat{s}_1, s_2) = c \mbox{ and } g(\hat{\hat{s}}_1, s_2)=a,
\]
or
\[
g(\hat{s}_1, s_2) = c \mbox{ and } g(\hat{\hat{s}}_1, s_2) = b.
\]
\end{claim}

\begin{proof}
Recall that we have assumed that for every strategy of agent $i$ there is some preference for which it is undominated. Suppose that $R_2$ ranks $a$ top. Then part (2) of Lemma \ref{lem:structure} implies that $U\!D_2(R_2) \subseteq S_2^a \cup S_2^b$. Analogously, if $R_2$ ranks $b$ top, then $U\!D_2(R_2) \subseteq S_2^a \cup S_2^b$. By part (1) of Lemma \ref{lem:structure}, any $s_2 \in U\!D_2(cab)$ satisfies:
	\[
	g(\hat{s}_1, s_2)=c \mbox{ and } g(\hat{\hat{s}}_1, s_2)=a,
	\]
	and any $s_2 \in U\!D_2(cba)$ satisfies:
	\[
	g(\hat{s}_1, s_2) =  c\mbox{ and } g(\hat{\hat{s}}_1, s_2)=b.
	\]
\end{proof}

Before we proceed with our analysis of agent 2's strategies, we observe that the conclusions of {\sc Claim} \ref{agent2strategies} allows us to narrow down the set of possible candidates for the preference $\hat R_1$.

\begin{claim}\label{1spreferences}
$\hat R_1$ is either $acb$ or $bca$.
\end{claim}

\begin{proof}
f $\hat R_1$ ranks $c$ top, then $\hat{s}_1$ would weakly dominate $\hat{\hat{s}}_1$, contradicting that $\hat{\hat s}_1 \in U\!D_1 (\hat R_1)$. If $\hat R_1$ ranks $c$ bottom, then $\hat{\hat{s}}_1$ would weakly dominate $\hat{s}_1$, contradicting that $\hat{s}_1 \in U\!D_1 (\hat R_1)$.
\end{proof}

Without loss of generality, we assume that $\hat R_1=bca$. We now return to our analysis of agent 2's strategy set.

\begin{claim}\label{agent2strategiesAgain}
 There are exactly two strategies in $S_2$ that are not in $S_2^a \cup S_2^b$. One of these, which we shall denote by $s_2^{cab}$, satisfies
\[
g(\hat{s}_1, s_2^{cab}) = c \mbox{ and } g(\hat{\hat{s}}_1, s_2^{cab})=a,
\]
and the other one, which we shall denote by $s_2^{cba}$, satisfies
\[
g(\hat{s}_1, s_2^{cba}) = c \mbox{ and } g(\hat{\hat{s}}_1, s_2^{cba}) = b.
\]
Moreover, $U\!D_2(cab)=\{s_2^{cab}\}$ and $U\!D_2(cba)=\{s_2^{cba}\}$.
\end{claim}

\begin{proof}
The argument in the proof of Claim \ref{agent2strategies} shows that it suffices to prove that $U\!D_2(cab)$ and $U\!D_2(cab)$ each have no more than one element. Without loss of generality we show this only for $U\!D_2(cab)$.  Suppose that $U\!D_2(cab)$ had more than one element.  By part (1) of Lemma \ref{lem:structure}, any $s_2 \in U\!D_2(cab)$ satisfies:
	\[
	g(\hat{s}_1, s_2)=c \mbox{ and } g(\hat{\hat{s}}_1, s_2)=a.
	\]
Now consider the preference pair consisting of $\hat R_1$ and of $cab$. We could apply to this preference profile the same reasoning as we applied above to the preference profile $\hat R_1$ and $\hat R_2$, with the roles of agents 1 and 2 swapped. We could infer, as we did above in {\sc Claim} \ref{1spreferences}, that agent 2's preference must be such that $b$ is ranked in the middle. But this contradicts that agent 2's preference is $cab$.
\end{proof}

What we have inferred so far allows us is symbolically represented by the middle panel in the top row of Figure \ref{fig:case1}. After we have pinned down the strategies that are not in $S_2^a\cup S_2^b$, we now return to the strategies of agent 2 that are in this set.

\begin{claim}\label{enforcingA}
$s_2\in S_2^a$ implies $g(s_1,s_2)=a$ for all $s_1\in S_1$. Moreover, $S_2^a$ has only one element, and $U\!D(abc)=U\!D(acb)=S_2^a$.
\end{claim}

\begin{proof}
	The second sentence is an immediate implication of the first sentence, the assumption that there are no duplicate strategies, and the definition of weak dominance. For an indirect proof of the first sentence, suppose that for some $s_2\in S_2^a$, we have $g(s_1, s_2) \neq a$ for some $s_1 \in S_1$. Then the preference $\hat R_1 = bca$ ranks $g(s_1,s_2)$ strictly above $a$. But then by Lemma \ref{lem:UD}, there would have to be a strategy $s_1' \in U\!D_1 (\hat R_1)$ such that $g(s_1', s_2)$ is ranked above $a$. This contradicts that $s_2 \in S_2^a$.
\end{proof}

The right panel in the top row of Figure \ref{fig:case1} symbolizes what we have concluded so far. Next, we can pin down the preference $\hat R_2$.

\begin{claim}\label{WhatIsR2?}
$\hat R_2 = bac$.
\end{claim}

\begin{proof} It cannot be that $\hat R_2$ ranks $a$ top, because then $\hat{s}_2$ would be a dominant strategy, and therefore would contradict with our assumption that $\hat R_2$ has at least two undominated strategies. It cannot be that $\hat R_2$ ranks $a$ bottom, because then $\hat{s}_2$ would be weakly dominated. Finally, it cannot be that $\hat R_2$ ranks $c$ top, because then, by Lemma \ref{lem:UD}, $U\!D_2 (\hat R_2)$ would have to include a strategy that yields $c$ against $\hat{s}_1$, which contradicts that $U\!D_2(\hat R_2) \subseteq S_2^a \cup S_2^b$. It follows that $\hat R_2 = bac$. 	
\end{proof}

\begin{claim}\label{Undominatedbac}
 $U\!D_2(bac)=\{\hat{s}_2,\hat{\hat{s}}_2\}$.
 \end{claim}

\begin{proof}
	From {\sc Claim} \ref{enforcingA}, we know that $\hat{s}_2$ is the unique element in $S_2^a$. We now show that $M_1(s_2)=\{a,b,c\}$ for all $s_2 \in S_2^b \cap U\!D_2 (\hat R_2)$. It then follows from Corollary \ref{cor:not-same-menu} that $\hat{\hat{s}}_2$ is the unique element in $S_2^b \cap U\!D_2(\hat R_2)$. The claim follows since $U\!D_2(\hat R_2) \subseteq S_2^a \cup S_2^b$.\medskip
	
	We proceed by elimination. It cannot be that $M_1 (s_2) = \{b\}$, nor that $M_1 (s_2) = \{a,b\}$, because in both cases $\hat{s}_2$ would be weakly dominated given $\hat R_2$. It remains to eliminate the possibility that $M_1 (s_2)=\{b,c\}$.\medskip
	
	Suppose that for some $s_2 \in S_2^b \cap U\!D_2 (\hat R_2)$, $M_1 (s_2) = \{b, c\}$. First consider the set of undominated strategies of agent 1 when she has preference $acb$. Part (1) of Lemma \ref{lem:structure} implies that $g(s_1, s_2) = c$ for all $s_1 \in U\!D_1(acb)$. Next we consider agent 2 when he has preference $cab$. Recall from {\sc Claim} \ref{agent2strategiesAgain} that agent 2 with this preference has a dominant strategy $s_2^{cab}$. We can then conclude that $g(s_1, s_2^{cab}) = c$ for all $s_1 \in U\!D_1(acb)$. But Lemma \ref{lem:UD}, combined with $g(\hat{\hat s}_1, s_2^{cab}) = a$, which we established in {\sc Claim} \ref{agent2strategiesAgain}, implies that there must exist some $s_1' \in U\!D_1(acb)$ such that $g(s_1', s_2^{cab}) = a$. We have thus obtained a contradiction, and the only remaining possibility is that $M_1 (s_2) = \{a,b,c\}$ for all $s_2 \in S_2^b \cap U\!D_2 (\hat R_2)$, which is what we wanted to show.
\end{proof}

By now, we know that agent 2, if he ranks $a$ top, has a dominant strategy $\hat{s}_2$. We also know that for every preference that ranks $c$ top, agent 2 has a dominant strategy, as described in {\sc Claim} \ref{agent2strategiesAgain}. Finally, we know that agent 2 with preference $bac$ has two undominated strategies: $\hat{s}_2$ and $\hat{\hat{s}}_2$. The left panel in the middle row of Figure \ref{fig:case1} symbolically represents what we have obtained so far. In the next step, we shall investigate agent 2's undominated strategies if he has preference $bca$.

\begin{claim}\label{Undominatedbca}
$|U\!D_2(bca)|=1$.
\end{claim}

\begin{proof}
	We first show that $U\!D_2(bca) \subseteq S_2^b$. By part (2) of Lemma \ref{lem:structure}, and by the results that we have so far obtained for agent 2's strategy set, we have to have: $U\!D_2(bca) \subseteq S_2^a \cup S_2^b$. If there exists a strategy $s_2 \in U\!D_2(bca)$ but $s_2 \notin S_2^b$, then what we have established so far implies that it must be the strategy $\hat{s}_2$. But $\hat{s}_2$ is weakly dominated if agent 2 has preference $bca$. Therefore, we conclude $U\!D_2(bca) \subseteq S_2^b$.\medskip
	
	Strategies in $U\!D_2(bca)$ cannot offer the menu $\{b\}$ or $\{a,b\}$, because then the strategy corresponding to this menu would weakly dominate $\hat{s}_2$ for agent $2$ with preference $bac$, which contracts with Claim \ref{WhatIsR2?}. Thus, strategies in $U\!D_2(bca)$ must either offer $\{b,c\}$ or $\{a,b,c\}$.\medskip
	
	Suppose that $U\!D_2(bca)$ has at least two elements. Then Corollary \ref{cor:not-same-menu} implies that there are exactly two strategies in $U\!D_2(bca)$, with one strategy offering the menu $\{b,c\}$ and the other strategy offering the menu $\{a,b,c\}$. In what follows, we show that this leads to a contradiction.\medskip
	
	First consider agent 1 with preference $acb$. By part (1) of Lemma \ref{lem:structure}, each of her undominated strategies $s_1 \in U\!D_1(acb)$ must satisfy (1) $g(s_1, s_2)=c$ if $s_2\in U\!D_2(bca)$ and $M_1(s_2)=\{b,c\}$; and (2) $g(s_1, s_2)=a$ if $s_2\in U\!D_2(bca)$ and $M_1(s_2)=\{a,b,c\}$. Now consider agent 2 with preference $cab$. Claim \ref{agent2strategiesAgain} showed that agent 2 with this preference has a dominant strategy $s_2^{cab}$. Because the strategy is dominant, we have to have: $g(s_1, s_2^{cab}) = c$ for all $s_1 \in U\!D_1(acb)$. Claim \ref{agent2strategiesAgain} also showed that $g(\hat{\hat{s}}_1, s_2^{cab}) = a$. But Lemma \ref{lem:UD} then implies that $g(s_1, s_2^{cab}) = a$ for at least one $s_1 \in U\!D_1(acb)$. We have found a contradiction.
\end{proof}

Since $|U\!D_2(bca)|=1$, agent 2 with preference $bca$ also has a dominant strategy. We denote this strategy by $s_2^{bca}$. Our discussion of agent 2's strategy set so far says that  agent 2 has either four (if $s_2^{bca}=\hat{\hat{s}}_2$) or five (if $s_2^{bca} \neq \hat{\hat{s}}_2$) strategies. We will resolve this question in the last step for Case 1. For the moment, we turn to agent 1's strategies.\medskip

\begin{claim}\label{bNotInMenu}
 For all $s_1\in S_1\setminus U\!D_1(bca)$ we have $b\notin M_2(s_1)$.
\end{claim}

\begin{proof}
	The proof is indirect. Suppose that there exists some $s_1 \in S_1 \setminus U\!D_1(bca)$ such that $b \in M_2(s_1)$. We are going to show that $s_1$ is a duplicate of one of the strategies in $U\!D_1(bca)$, which contradicts our assumption that there are no duplicate strategies. We distinguish two cases. The first is that $M_2(s_1) = \{a,b\}$, and the second case is that $M_2(s_1) = \{a,b,c\}$. The arguments for the two cases are completely analogous. Therefore, here we only deal with the case that $M_2(s_1) = \{a,b\}$. Applying Lemma \ref{lem:UD} to agent 2 with preference $bac$, we can conclude that $g(s_1, \hat{\hat{s}}_2) = b$. Because for all other preferences agent 2 has dominant strategies that we have already identified, we can conclude that:
	\[g(s_1, \hat{s}_2) = g(s_1,s_2^{cab})= a \mbox{ and } g(s_1, s_2^{cba}) = g(s_1,s_2^{bca}) = b.\]
	This implies that $s_1$ is a duplicate strategy of $\hat{\hat{s}}_1$.
\end{proof}

This claim implies that strategies that are not in $U\!D_1(bca)$ must yield either $a$ or $c$ against any other strategy of agent 2. Let us focus on the alternative that they yield when agent 2 chooses $\hat{\hat{s}}_2$. The next two claims show that there is only one strategy outside of $U\!D_1(bca)$ that yields $c$ against $\hat{\hat{s}}_2$, and also only one such strategy that yields $a$ against $\hat{\hat{s}}_2$. This then implies that agent 1 has only four strategies, the two strategies in $U\!D_1(bca)$, and the two strategies not in $U\!D_1(bca)$.

\begin{claim}\label{Uniquec}
 There is a unique strategy $s_1\in S_1\setminus U\!D_1(bca)$ such that $g(s_1, \hat{\hat{s}}_2) = c$. Furthermore, for this strategy we have:
\[
g(s_1, s_2^{bca}) = g(s_1, s_2^{cab}) = g(s_1,s_2^{cba})=c.
\]
\end{claim}

\begin{proof}
Recall that in the proof of {\sc Claim} \ref{Undominatedbac}, we concluded that $M_1(\hat{\hat{s}}_2)=\{a,b,c\}$. This implies that there is at least one strategy $s_1$ such that $g(s_1, \hat{\hat{s}}_2) = c$.
From Claims \ref{enforcingA} and  \ref{bNotInMenu}, we know that $M_2(s_1) = \{a,c\}$. Because we already know that agent 2 with preferences $bca$, $cab$, or $cba$ has dominant strategies, we know that
$g(s_1, s_2^{bca}) = g(s_1, s_2^{cab}) = g(s_1,s_2^{cba})=c$. We have now pinned down for all strategies of agent 2 which outcome results if agent 1 chooses a strategy $s_1 \in S_1 \setminus U\!D_1(bca)$ such that $g(s_1, \hat{\hat{s}}_2) = c$. The uniqueness of such a strategy is therefore a consequence of the assumption that there are no duplicate strategies.
\end{proof}

\begin{claim}\label{Uniquea}
There is a unique strategy $s_1$ such that $g(s_1, \hat{\hat{s}}_2) = a$. Furthermore, for this strategy we have:
\[
g(s_1, s_2^{bca}) = g(s_1,s_2^{cab}) = g(s_1,s_2^{cba}) = a.
\]
\end{claim}

\begin{proof}
Recall that in the proof of {\sc Claim} \ref{Undominatedbac}, we concluded that $M_1(\hat{\hat{s}}_2)=\{a,b,c\}$. This implies that there is at least one strategy $s_1$ such that $g(s_1, \hat{\hat{s}}_2) = a$. From Claim \ref{bNotInMenu} we can then infer that: $M_2(s_1)$ is either $\{a\}$ or $\{a,c\}$. For ease of notation, let:
	\[
	S_1^a = \{s_1 \in S_1: g(s_1, \hat{s}_2)=g(s_1, \hat{\hat{s}}_2)=a\}.
	\]
	
	We first show that there exists at least one strategy $s_1 \in S_1^a$ that offers the menu $\{a\}$. The proof is indirect. Suppose that $M_2(s_1)=\{a,c\}$ for all $s_1 \in S_1^a$. We must have:
	\[
	g(s_1, s_2^{bca}) = g(s_1,s_2^{cab}) = g(s_1,s_2^{cba})=c
	\]
	for all $s_1 \in S_1^a$. This is because all the strategies of agent 2 that we are referring to are dominant strategies. Because there are no duplicate strategies, we obtain that there is a unique element $s_1$ in $S_1^a$, and that for this strategy
	\[
	g(s_1, s_2^{bca}) = g(s_1,s_2^{cab}) = g(s_1,s_2^{cba})=c.
	\]
	Now consider agent 1 who ranks $a$ top. The unique element in $S_1^a$ cannot be weakly dominated, because this is the only strategy that yields outcome $a$ against strategy $\hat{\hat s}_2$. But since $g(\hat{\hat s}_1, s_2^{cab}) = a$, by Lemma \ref{lem:UD}, she must have another undominated strategy that yields $a$ against $s_2^{cab}$. But then, if agent 1 has a preference that ranks $a$ top, and agent 2 has preference $bac$, there is no local dictator. Thus we have obtained a contradiction.\medskip
	
	Therefore, there must exist at least one strategy $s_1 \in S_1^a$ such that $M_2(s_1)=\{a\}$. Because there are no duplicate strategies, there can only be one such strategy. But now suppose there is also a strategy $s_1' \in S_1^a$ with $M_2(s_1) = \{a,c\}$. As before, it follows that
	\[
	g(s_1', s_2^{cba}) = g(s_1', s_2^{bca}) = g(s_1', s_2^{cab}) = c.
	\]
	But note that $s_1'$ cannot be undominated for any preference, and we have ruled out that strategies that are not dominated for all preferences are included in the mechanism. The claim follows.
\end{proof}

What we have found so far establishes that agent 1 has four strategies and agent 2 has either four (if $s_2^{bca} = \hat{\hat{s}}_2$) or five (if $s_2^{bca} \neq \hat{\hat{s}}_2$) strategies. Moreover, for any strategy combination, we know which outcome results. If agent 2 has five strategies, then the mechanism must take the form shown in the left panel in the bottom row of Figure \ref{fig:case1}. But note that in that panel $\hat{\hat {s}}_2$ and $s_2^{bca}$ are duplicate strategies. Because we have assumed that there are no duplicate strategies, we can conclude that agent 2 has four strategies and the mechanism is the one shown in the right panel in the bottom row of Figure \ref{fig:case1}. This completes the proof for Case 1. \medskip

\noindent{\sc Case 2:} For all preference profiles such that both agents have multiple undominated strategies, exactly one outcome may result if agents choose from the strategies that are undominated for these preference profiles. \medskip

Let us denote by $(\tilde R_1, \tilde R_2)$ a preference profile for which both agents have more than one undominated strategies. Without loss of generality, let us assume that $g(U\!D_1(\tilde R_1), U\!D_2(\tilde R_2)) = \{a\}$.\medskip

Figure \ref{fig:case2} illustrates the proof for the second case. We shall refer to Figure \ref{fig:case2} while presenting the proof. The left panel in the top row shows the starting point of the proof. We begin with an analysis of the sets $U\!D_i(\tilde R_i)$ for each agent and of the menus offered by the strategies in these sets.

\begin{figure}[h]
	\centering
	\begin{tikzpicture}[scale=0.65]
	
	\subfloat[Step (b1)]{
		\node [left] at (0, 5) {$a$};
		\node [left] at (1.3, 5) {$\cdots$};
		\node [left] at (2, 5) {$a$};
		\node[rotate=90,yshift=0pt] at (-0.4, 4) {$\cdots$};
		\node[rotate=90,yshift=0pt] at (1.6, 4) {$\cdots$};
		\node [left] at (0, 3) {$a$};
		\node [left] at (1.3, 3) {$\cdots$};
		\node [left] at (2, 3) {$a$};
		\draw [decorate,decoration={brace,amplitude=7pt},xshift=0pt,yshift=0pt]
		(-1.3, 2.5) -- (-1.3, 5.5) node [black,midway,xshift=-1.1cm]
		{\footnotesize $U\!D_1 (\tilde R_1)$};
		\draw [decorate,decoration={brace,amplitude=7pt},xshift=0pt,yshift=0pt]
		(-1,5.7) -- (2,5.7) node [black,midway,yshift=0.5cm]
		{\footnotesize $U\!D_2 (\tilde R_2)$};
	}
	
	\node [left] at (4, 4.5) {$\Rightarrow$};
	
	\subfloat[Step (b2)]{
		\node [left] at (8, 5) {$\tilde{s}_1$};
		\node [left] at (8, 4) {$\tilde{\tilde{s}}_1$};
		\node [left] at (8.2, 3) {$s_1^{cab}$};
		\node [left] at (8.2, 2) {$s_1^{cba}$};
		\node [left] at (9.1, 6) {$\tilde{s}_2$};
		\node [left] at (10.2, 6) {$\tilde{\tilde{s}}_2$};
		\node [left] at (11.7, 6) {$s_2^{cab}$};
		\node [left] at (13, 6) {$s_2^{cba}$};
		\node [left] at (9, 5) {$a$};
		\node [left] at (9, 4) {$a$};
		\node [left] at (9, 3) {$c$};
		\node [left] at (9, 2) {$c$};
		\node [left] at (10, 5) {$a$};
		\node [left] at (10, 4) {$a$};
		\node [left] at (10, 3) {$a$};
		\node [left] at (10, 2) {$b$};
		\node [left] at (11.2, 5) {$c$};
		\node [left] at (11.2, 4) {$a$};
		\node [left] at (12.3, 5) {$c$};
		\node [left] at (12.3, 4) {$b$};
		\draw [decorate,decoration={brace,amplitude=7pt},xshift=0pt,yshift=0pt]
		(7, 3.5) -- (7, 5.5) node [black,midway,xshift=-1.1cm]
		{\footnotesize $U\!D_1 (\tilde R_1)$};
		\draw [decorate,decoration={brace,amplitude=7pt},xshift=0pt,yshift=0pt]
		(8,6.5) -- (10,6.5) node [black,midway,yshift=0.5cm]
		{\footnotesize $U\!D_2 (\tilde R_2)$};
	}
	
	\node [left] at (-6, -3.5) {$\Rightarrow$};
	
	\subfloat[Step (b3)]{
		\node [left] at (-2, -3) {$\tilde{s}_1$};
		\node [left] at (-2, -4) {$\tilde{\tilde{s}}_1$};
		\node [left] at (-2, -5) {$s_1^b$};
		\node [left] at (-1.8, -6) {$s_1^{cab}$};
		\node [left] at (-1.8, -7) {$s_1^{cba}$};
		\node [left] at (-0.9, -2) {$\tilde{s}_2$};
		\node [left] at (0.2, -2) {$\tilde{\tilde{s}}_2$};
		\node [left] at (1.3, -2) {$s_2^b$};
		\node [left] at (2.8, -2) {$s_2^{cab}$};
		\node [left] at (4.1, -2) {$s_2^{cba}$};
		
		\node [left] at (-1, -3) {$a$};
		\node [left] at (0, -3) {$a$};
		\node [left] at (1.2, -3) {$b$};
		\node [left] at (2.3, -3) {$c$};
		\node [left] at (3.4, -3) {$c$};
		
		\node [left] at (-1, -4) {$a$};
		\node [left] at (0, -4) {$a$};
		\node [left] at (1.2, -4) {$b$};
		\node [left] at (2.3, -4) {$a$};
		\node [left] at (3.4, -4) {$b$};
		
		\node [left] at (-1, -5) {$b$};
		\node [left] at (0, -5) {$b$};
		\node [left] at (1.2, -5) {$b$};
		\node [left] at (2.3, -5) {$b$};
		\node [left] at (3.4, -5) {$b$};
		
		\node [left] at (-1, -6) {$c$};
		\node [left] at (0, -6) {$a$};
		\node [left] at (1.2, -6) {$b$};
		
		\node [left] at (-1, -7) {$c$};
		\node [left] at (0, -7) {$b$};
		\node [left] at (1.2, -7) {$b$};
		
		\draw [decorate,decoration={brace,amplitude=7pt},xshift=0pt,yshift=0pt]
		(-3, -4.5) -- (-3, -2.5) node [black,midway,xshift=-1.1cm]
		{\footnotesize $U\!D_1(acb)$};
		\draw [decorate,decoration={brace,amplitude=7pt},xshift=0pt,yshift=0pt]
		(-2,-1.5) -- (0,-1.5) node [black,midway,yshift=0.5cm]
		{\footnotesize $U\!D_2 (acb)$};
	}
	
	\node [left] at (5, -3.5) {$\Rightarrow$};
	
	\subfloat[Step (b4)]{
		\node [left] at (9, -3) {$\tilde{s}_1$};
		\node [left] at (9, -4) {$\tilde{\tilde{s}}_1$};
		\node [left] at (9, -5) {$s_1^b$};
		\node [left] at (9.2, -6) {$s_1^{cab}$};
		\node [left] at (9.2, -7) {$s_1^{cba}$};
		\node [left] at (10.1, -2) {$\tilde{s}_2$};
		\node [left] at (11.2, -2) {$\tilde{\tilde{s}}_2$};
		\node [left] at (12.3, -2) {$s_2^b$};
		\node [left] at (13.8, -2) {$s_2^{cab}$};
		\node [left] at (15.1, -2) {$s_2^{cba}$};
		
		\node [left] at (10, -3) {$a$};
		\node [left] at (11, -3) {$a$};
		\node [left] at (12.2, -3) {$b$};
		\node [left] at (13.3, -3) {$c$};
		\node [left] at (14.4, -3) {$c$};
		
		\node [left] at (10, -4) {$a$};
		\node [left] at (11, -4) {$a$};
		\node [left] at (12.2, -4) {$b$};
		\node [left] at (13.3, -4) {$a$};
		\node [left] at (14.4, -4) {$b$};
		
		\node [left] at (10, -5) {$b$};
		\node [left] at (11, -5) {$b$};
		\node [left] at (12.2, -5) {$b$};
		\node [left] at (13.3, -5) {$b$};
		\node [left] at (14.4, -5) {$b$};
		
		\node [left] at (10, -6) {$c$};
		\node [left] at (11, -6) {$a$};
		\node [left] at (12.2, -6) {$b$};
		\node [left] at (13.3, -6) {$c$};
		\node [left] at (14.4, -6) {$c$};
		
		\node [left] at (10, -7) {$c$};
		\node [left] at (11, -7) {$b$};
		\node [left] at (12.2, -7) {$b$};
		\node [left] at (13.3, -7) {$c$};
		\node [left] at (14.4, -7) {$c$};
		
		\draw [decorate,decoration={brace,amplitude=7pt},xshift=0pt,yshift=0pt]
		(8, -4.5) -- (8, -2.5) node [black,midway,xshift=-1.1cm]
		{\footnotesize $U\!D_1(acb)$};
		\draw [decorate,decoration={brace,amplitude=7pt},xshift=0pt,yshift=0pt]
		(9,-1.5) -- (11,-1.5) node [black,midway,yshift=0.5cm]
		{\footnotesize $U\!D_2 (acb)$};
	}
	
	\end{tikzpicture}
	\caption{\small There is a unique type 2 strategically simple mechanism (up to relabeling) in the second case.} \label{fig:case2}
\end{figure}

\begin{claim}\label{NotUniquea}
If $s_1 \in U\!D_1(\tilde R_1)$ then $M_2 (s_1) \neq \{a\}$. (The analogous statement for agent 2 can be proved in the same way.)
\end{claim}

\begin{proof}
	The proof is indirect. Suppose that $M_2 (s_1) = \{a\}$ for some $s_1 \in U\!D_1 (\tilde R_1)$. Let $s_1'$ be another element of $U\!D_1(\tilde R_1)$. First observe that $M_2 (s_1')$ has to be $\{a,b,c\}$, because in all other cases, for every preference of agent 1, either $s_1$ weakly dominates $s_1'$ or the other way round.\medskip
	
	For both $s_1$ and $s_1'$ to be undominated for agent 1 with preference $\tilde R_1$, it must be that $\tilde R_1$ ranks $a$ in the middle. Without loss of generality, we assume that $\tilde R_1=bac$. By Lemma \ref{lem:UD}, we conclude that $b \notin M_1(s_2)$ for any $s_2 \in U\!D_2(\tilde R_2)$.\medskip
	
	Now let $s_2$ and $s_2'$ denote two different elements of $U\!D_2(\tilde R_2).$ We just concluded that neither strategy offers a menu that includes $b$. By Corollary 1, they have to offer different menus, and therefore, without loss of generality, we can write that $M_1(s_2) = \{a\}$ and $M_1(s_2') = \{a,c\}$. But then there is no preference of agent 2 under which both $s_2$ and $s_2'$ are undominated.
\end{proof}

\begin{claim}\label{Case2Menus}
The set $U\!D_1(\tilde R_1)$ has exactly two elements, say $\tilde s_1$ and $\tilde{\tilde s}_1$. Moreover, for one of these two strategies, say $\tilde s_1$, we have: $M_2 (\tilde s_1) = \{a, b, c\}$. For the other strategy, either $M_2(\tilde{\tilde s}_1)=\{a,b\}$ or $M_2(\tilde{\tilde s}_1)=\{a,c\}$. (The analogous claim is true for agent 2.)
\end{claim}

\begin{proof}
The claim follows from Claim \ref{NotUniquea} and Corollary \ref{cor:not-same-menu} once we rule out the case in which that there are simultaneously a strategy in $U\!D_1(\tilde R_1)$ that offers menu $\{a,b\}$ and another strategy in $U\!D_1(\tilde R_1)$ that offers menu $\{a,c\}$. We prove indirectly that this cannot be the case.\medskip
	
	Thus we assume that there is a strategy $s_1 \in U\!D_1(\tilde R_1)$ with $M_2(s_1)=\{a,b\}$ and another strategy $s_1' \in U\!D_1(\tilde R_1)$ with $M_2(s_1') = \{a, c\}$. By Lemma \ref{lem:UD}, it would have to be the case that for agent 2 with preference $bca$, there is an undominated strategy that yields $b$ against $s_1$ and also an undominated strategy that yields $c$ against $s_1'$. Therefore, we would conclude that $g(U\!D_1(\tilde R_1), U\!D_2(bca))=2$. By the definition of case 2, it has to be that $U\!D_2(bca)$ has just one element. In other words, agent 2 with preference $bca$ has a dominant strategy $s_2^{bca}$. Using the same arguments as above, we can conclude that agent 2 with preference $cba$ has a dominant strategy $s_2^{cba}$.\medskip
	
	Now consider any two different strategies $s_2, s_2' \in U\!D_2(\tilde R_2)$. By assumption, both strategies' menus include $a$, and by Claim \ref{NotUniquea} cannot only include $a$. Therefore, at least one of these menus must contain exactly two elements, one of which is $a$. Without loss of generality let the other one be $c$. Thus, we consider: $M_1 (s_2) = \{a, c\}$. By Corollary \ref{cor:not-same-menu}, $s_2'$ has to offer a different menu, and this implies: $b \in M_1 (s_2').$\medskip
	
	Using the same argument as in the second paragraph of the current proof, we can conclude that agent 1 with preference $bca$ has a dominant strategy, say $s_1^{bca}$, and that $g(s_1^{bca}, s_2) = c$, and that $g(s_1^{bca}, s_2') = b$.\medskip
	
	Now consider $g(s_1^{bca}, s_2^{cba}).$ Because $s_1^{bca}$ is a dominant strategy for agent 1 with preference $bca$, and because $g(s_1, s_2^{cba}) = b$, it follows that
	$g(s_1^{bca}, s_2^{cba}) = b$. But similarly, because $s_2^{cba}$ is a dominant strategy for agent 2 with preference $cba$, and because $g(s_1^{bca}, s_2) = c$, it follows that  $g(s_1^{bca}, s_2^{cba}) = c$. We have obtained a contradiction.
\end{proof}

Without loss of generality, we now assume that $M_2 (\tilde{\tilde s}_1)=\{a,b\}$. Next, we show that, as a consequence, we have to have that $M_1(\tilde{\tilde s}_2)=\{a,b\}$.

\begin{claim}\label{TheSecondMenu}
$M_1 (\tilde{\tilde s}_2)=\{a,b\}$.
\end{claim}

\begin{proof}
	The proof is indirect. Suppose that $M_1 (\tilde{\tilde s}_2) = \{a, c\}$. As in the proof of claim \ref{Case2Menus}, we can then infer that agent 1 with preference $bca$ has a dominant strategy $s_1^{bca}$. Furthermore, $g(s_1^{bca}, \tilde s_2) = b$ and $g(s_1^{bca}, \tilde{\tilde s}_2) = c$. Similarly, the assumption that $M_2 (\tilde{\tilde s}_1)=\{a,b\}$ implies that agent 2 with preference $cba$ has a dominant strategy $s_2^{cba}$. Furthermore, $g(\tilde s_1, s_2^{cba})=c$ and $g(\tilde{\tilde s}_1, s_2^{cba}) = b$. A contradiction is then reached as in the proof of claim \ref{Case2Menus} by showing that $g(s_1^{bca}, s_2^{cba})$ has to be simultaneously $b$ and $c$.
\end{proof}

\begin{claim}\label{CDominant}
 Agent 1 with preference $cab$ has a dominant strategy $s_1^{cab}$, and $g(s_1^{cab}, \tilde s_2) = c$ and $g(s_1^{cab}, \tilde{\tilde s}_2) = a$. Agent 1 with preference $cba$ has a dominant strategy $s_1^{cba}$, and $g(s_1^{cba}, \tilde s_2) = c$ and $g(s_1^{cba}, \tilde{\tilde s}_2) = b$.  (The analogous claims are true for agent 2.)
\end{claim}

\begin{proof}
This follows from the arguments used in the second paragraph of the proof of claim \ref{Case2Menus}.
\end{proof}

At this point we have a good understanding of the sets $U\!D_i(\tilde R_i)$ and of the menus offered by the strategies in these sets. What we have obtained so far is symbolically represented in the right panel in the top row in Figure \ref{fig:case2}. (Observe that the strategies $s_i^{cab}$ and $s_i^{cba}$ are not contained in $U\!D(\tilde R_i)$.)

\begin{claim}\label{Case2Obvious}
 If agent 1 ranks $a$ top, then every undominated strategy $s_1$ of agent 1 satisfies $g(s_1, \tilde s_2) = g(s_1, \tilde{\tilde s}_2) = a$. If agent 1 ranks $b$ top, then every undominated strategy $s_1$ of agent 1 satisfies $g(s_1, \tilde s_2) = g(s_1, \tilde{\tilde s}_2) = b$. (The analogous claims are true for agent 2.)
\end{claim}

\begin{proof}
	This follows from part (2) of Lemma \ref{lem:structure} and from the definition of Case 2.	
\end{proof}

\begin{claim}\label{Case2Preferences}
 $\tilde R_1=\tilde R_2=acb$.
\end{claim}

\begin{proof}
Claims \ref{CDominant} and \ref{Case2Obvious}  show that an agent with multiple undominated strategies must rank $a$ top. This leaves just two possible preferences: $abc$ and $acb$. But if $\tilde R_1 = abc$, then clearly $\tilde{\tilde s}_1$ would weakly dominate $\tilde s_1$.
\end{proof}

\begin{claim}\label{BDominant}
 There is a unique strategy, say $s_1^b$, such that, if agent 1 ranks $b$ top, then this strategy is dominant. Moreover, $g(s_1^b,s_2)=b$ for all $s_2 \in S_2$. (The analogous statement is true for agent 2.)
\end{claim}

\begin{proof}
	By Claim \ref{Case2Obvious}, if agent 1 ranks $b$ top, every undominated strategy $s_1$ of agent 1 satisfies: $g(s_1, \tilde s_2) = g(s_1, \tilde{\tilde s}_2) = b$. Claim \ref{Case2Preferences} showed that $\tilde R_2 = acb$, which ranks $b$ bottom. Therefore, by Lemma \ref{lem:UD}, we have to have that $g(s_1, s_2) = b$ for all $s_2 \in S_2$. There can only be one such strategy, because there are no duplicate strategies. Moreover, this strategy is dominant whenever agent 1 ranks $b$ top.
\end{proof}

The left panel in the bottom row of Figure \ref{fig:case2} shows what we have inferred so far about the mechanism.

\begin{claim}\label{ADominant}
 For agent 1 with preference $abc$, strategy $\tilde{\tilde s}_1$ is dominant. (The analogous statement is true for agent 2.)
\end{claim}

\begin{proof}
		Consider agent 1 with preference $abc$. Whenever agent 2's strategy is undominated for a preference that puts $a$ top, then, by Claim \ref{Case2Obvious}, if agent 1 chooses $\tilde{\tilde s}_1$, the outcome is $a$, which is agent 1's most preferred outcome. If agent 2's strategy is undominated for a preference that puts $b$ at the top, by Claim \ref{BDominant}, all strategies of agent 1 yield the same outcome $b$. Finally, if agent 2 chooses an undominated strategy for preference $cab$, then, by Claim \ref{CDominant}, the outcome that results if agent 1 chooses $\tilde{\tilde s}_1$ is $a$.\medskip
	
	The only remaining case is that agent 2 has preference $cba$ and chooses his dominant strategy $s_2^{cba}$. By claim \ref{CDominant}, $g(\tilde{\tilde s}_1, s_2^{cba}) = b$. Thus we have to show that $a \notin M_1 (s_2^{cba}).$ If $a \in M_1(s_2^{cba})$, by Lemma \ref{lem:UD}, there would have to be an undominated strategy of agent 1 with preference $acb$ that yields $a$ against $s_2^{cba}$. In Claim \ref{CDominant} we showed that no such strategy exists.
\end{proof}

We can now wrap up the analysis of the second case. For five of the six possible preferences of each agent, we have established that they have dominant strategies. Moreover, for agents with preference $acb$, we have established that they have only two undominated strategies. Moreover, the dominant strategy of agents with preference $abc$ is one of the undominated strategies of agents with preference $acb$, and agents with preferences that put $b$ top have the same dominant strategy. A short calculation reveals that every agent has exactly 5 strategies. The results that we have obtained so far show for most strategy combinations which outcome results. What remains to be filled in is are the outcomes that result when both agents choose their strategies $s_i^{cab}$ and $s_i^{cba}$. But because these are dominant strategies, and because we already know that each agent has a strategy available that achieves outcome $c$ against these two strategies of the other agent, it must be that:
\[
g(s_1^{cab}, s_2^{cab}) = g(s_1^{cab}, s_2^{cba}) = g(s_1^{cba}, s_2^{cab}) = g(s_1^{cba}, s_2^{cba}) = c.
\]
Thus, there is a unique type 2 strategically simple mechanism (up to relabeling) in the second case as shown in the right panel in the bottom row of Figure \ref{fig:case2}.

\end{document}